\documentclass[11pt]{article}
\usepackage{amsmath, amssymb, amsfonts, amsthm}
\usepackage{epsfig}

\newtheorem{theorem}{Theorem}[section]
\newtheorem{proposition}[theorem]{Proposition}
\newtheorem{corollary}[theorem]{Corollary}
\newtheorem{lemma}[theorem]{Lemma}
\newtheorem{definition}[theorem]{Definition}

\newcommand{\rd}{{\rm d}}
\newcommand{\be}{\begin{equation}}
\newcommand{\ee}{\end{equation}}
\newcommand{\bey}{\begin{eqnarray}}
\newcommand{\eey}{\end{eqnarray}}

\newcommand{\E}{{\mathbb E }}
\renewcommand{\P}{{\mathbb P}}

\newcommand{\eps}{\varepsilon}

\newcommand{\bu}{{\bf u}}
\newcommand{\bba}{{\bf a}}
\newcommand{\bb}{{\bf b}}

\newcommand{\bz}{{\bf z}}

\newcommand{\e}{\varepsilon}

\newcommand{\bR}{{\mathbb R}}
\newcommand{\bC}{{\mathbb C}}
\newcommand{\bN}{{\mathbb N}}

\newcommand{\tr}{\mbox{Tr}}

\newcommand{\wt}{\widetilde}

\newcommand{\const}{\mathrm{const}}

\newcommand{\cN}{{\cal N}}

\input epsf

\newcommand{\donothing}[1]{}

\begin{document}

\title{Average Density of States for Hermitian Wigner Matrices}

\author{Anna Maltsev and Benjamin Schlein\thanks{Partially supported
by an ERC Starting Grant} \\
\\
Institute of Applied Mathematics, University of Bonn, \\
Endenicher Allee 60, 53115 Bonn}

\maketitle

\begin{abstract}
We consider ensembles of $N \times N$ Hermitian Wigner matrices, whose entries are (up to the symmetry constraints) independent and identically distributed random variables. Assuming sufficient regularity for the probability density function of the entries, we show that the expectation of the density of states on {\it arbitrarily} small intervals converges to the semicircle law, as $N$ tends to infinity.
\end{abstract}

\section{Introduction}
\setcounter{equation}{0}
\label{sec:intro}

Wigner matrices are matrices whose entries are independent and identically distributed random variables, up to symmetry constraints (one distinguishes between ensemble of real symmetric, Hermitian, and quaternion Hermitian Wigner matrices). They were first introduced by Wigner to describe the excitation spectra of heavy nuclei. Wigner's basic idea was as follows; the entries of the Hamilton operator of a complex system (such as a heavy nucleus) depend on too many degrees of freedom to be written down precisely. Hence, it makes sense to assume the entries of the Hamilton operator to be random variables, and to look for results which hold for most realizations of the randomness.

\medskip

Wigner's idea was very successful and, to this day, it is one of the most useful tools in nuclear physics. Since then, Wigner matrices have been linked to several different branches of mathematics and physics. Eigenvalues of random Schr\"odinger operators in the metallic phase are expected to share many similarities with eigenvalues of Hermitian Wigner matrices. Eigenvalues of the Laplace operators over a domain $\Omega \subset \bR^n$ with chaotic classical trajectories are expected to exhibit the same correlations as the eigenvalues of real symmetric Wigner matrices.
The zeros of Riemann's zeta function on the line ${\text Re } \, z = 1/2$ should be distributed, after appropriate rescaling, like eigenvalues of Hermitian Wigner matrices. And more examples are available.

\medskip

The success of Wigner's idea, and the variety of links of Wigner matrices to what appear to be completely unrelated branches of mathematics and physics  is a consequence of the phenomenon of {\it universality}; in vague terms, universality states that the statistical properties of the spectrum of matrices (or operators) with disorder (randomness) depend on the symmetries of the model under consideration, but otherwise they are largely independent of the details of the disorder.

\medskip

Within the realm of Wigner matrices, universality has a much more precise meaning. It refers to the fact that the local eigenvalue statistics (the local correlation functions) depend on the symmetry of the ensemble (real symmetric matrices, Hermitian matrices, and quaternion Hermitian matrices lead to different statistics), but they are otherwise independent of the particular choice of the probability law of the entries of the matrix. While universality at the edges of the spectrum (universality of the distribution of the largest, or the smallest, few eigenvalues) has been known since \cite{S}, universality in the bulk of the spectrum has been understood only recently; see \cite{EPRSY,TV,ERSTVY,ESY4,ESYY}.

\medskip

Let us now define the ensembles that we are going to consider more precisely. We focus here on ensembles of Hermitian Wigner matrices.
\begin{definition}\label{def} An ensemble of {\it Hermitian Wigner matrices} consists of $N \times N$ matrices $H = (h_{jk})_{1\leq j,k \leq N}$, with \[ \begin{split}  h_{jk} &= \frac{1}{\sqrt{N}} (x_{jk} + i y_{jk})  \qquad \text{for } 1 \leq j <k \leq N \\
h_{jk} & = \overline{h}_{kj} \qquad \text{for } 1\leq k < j \leq N \\
h_{jj} & = \frac{1}{\sqrt{N}} x_{jj} \qquad \text{for } 1 \leq j \leq N \end{split} \]
where $\{ x_{jk}, y_{jk}, x_{jj} \}_{1 \leq j < k \leq N}$ is a collection of $N^2$
independent real random variables. The (real and imaginary parts of the) off-diagonal entries $\{ x_{jk} , y_{jk} \}_{1\leq j < k \leq N}$ have a common distribution with \[ \E \, x_{jk} = 0 \qquad \text{and } \E \, x_{jk}^2 = \frac{1}{2} \, .  \]
Also the diagonal entries $\{ x_{jj} \}_{j=1}^N$ have a common distribution with
\[ \E \, x_{jj} = 0  \qquad \text{and } \E \, x_{jj}^2 = 1 \, . \]
\end{definition}
The scaling of the entries with the dimension $N$ ($h_{jk}$ is of the order $N^{-1/2}$) guarantees that, in the limit $N \to \infty$, all eigenvalues of $H$ remain of order one. In fact, it turns out that, as $N \to \infty$, all eigenvalues of $H$ are contained in the interval $[-2,2]$. In \cite{W}, Wigner showed the convergence of the density of states (density of eigenvalues) for ensembles of Wigner matrices to the famous semicircle law $\rho_{sc}$.
For arbitrary fixed $a \leq b$ and $\delta >0$, Wigner proved that
\begin{equation}\label{eq:Wsc} \lim_{N \to \infty} \P \left( \left| \frac{\cN [a;b]}{N(b-a)} - \frac{1}{(b-a)} \int_a^b \rd s \rho_{sc} (s) \right| \geq \delta \right) = 0 \end{equation}
where $\cN [a;b]$ denotes (here and henceforth) the number of eigenvalues in the interval $[a;b]$, and \begin{equation}\label{eq:sc} \rho_{sc} (E) = \left\{ \begin{array}{ll} \frac{1}{2\pi} \sqrt{1- \frac{E^2}{4}}, \quad &\text{if } |E| \leq 2 \\ 0 \quad &\text{if } |E| > 2 \end{array} \right. \, . \end{equation}
Note that the semicircle law is independent of the choice of the probability law for the entries of the matrices.

\medskip

An important special example of an ensemble of Hermitian Wigner matrices is the Gaussian Unitary Ensemble (GUE). It is characterized by the fact that (the real and imaginary parts of) all entries are Gaussian random variables, and it is the only ensemble of Hermitian Wigner matrices which is invariant with respect to unitary conjugation. If $H$ is a GUE matrix, and $U$ is an arbitrary fixed unitary matrix, then $UHU^*$ is again a GUE matrix (whose entries have exactly the same distribution as the entries of $H$). Because of the unitary invariance, for GUE it is possible to compute explicitly the joint probability density function for the $N$ eigenvalues. It is given by
\begin{equation} \label{eq:pGUE} p_{\text{GUE}} (\mu_1, \dots , \mu_N ) = \const \cdot \prod_{i<j}^N (\mu_i - \mu_j)^2 \, e^{-\frac{N}{2} \sum_{j=1}^N \mu_j^2}\,. \end{equation}
Here we think of $p_{\text{GUE}}$ as a probability density on $\bR^N$; there is no ordering among the variables $(\mu_1, \dots , \mu_N)$.  Starting from $p_{\text{GUE}}$ we define, for arbitrary $k=1, \dots , N$, the $k$-point correlation function
\[ p^{(k)}_{\text{GUE}} (\mu_1, \dots , \mu_k) = \int d\mu_{k+1} \dots d\mu_N \, p_{\text{GUE}} (\mu_1, \dots , \mu_N) \,.\]
Using the explicit expression (\ref{eq:pGUE}), Dyson was able to compute the local correlation functions of GUE in the limit $N \to \infty$. In \cite{D}, he proved that, for every $k \geq 1$,
\begin{equation}\label{eq:WD} \frac{1}{\rho^k_{sc} (E)} p^{(k)}_{\text{GUE}} \left(E + \frac{x_1}{\rho_{sc} (E) N} , \dots , E + \frac{x_k}{\rho_{sc} (E) N} \right)  \to \det \, \left( \frac{\sin (\pi (x_j - x_\ell))}{(\pi (x_j - x_\ell))} \right)_{1 \leq j ,\ell \leq k} \end{equation} as $N \to \infty$. The r.h.s. of (\ref{eq:WD}) is known as the Wigner-Dyson (or sine-kernel) distribution. Observe that the arguments of $p^{(k)}_{\text{GUE}}$ in (\ref{eq:WD}) vary within an interval of size of the order $1/N$ (hence the name of {\it local} correlations). Since the typical distance between eigenvalues is of the order $1/N$, it is not surprising that non-trivial correlations are observed on this scale.

\medskip

Dyson's proof of (\ref{eq:WD}) was based on the explicit expression (\ref{eq:pGUE}) for the joint probability density function of the eigenvalues of GUE matrices. GUE is the only ensemble of Hermitian Wigner matrices which enjoys unitary invariance; for this reason, it is the only ensemble of Hermitian Wigner matrices for which an explicit expression for the joint probability density function of the eigenvalues exists. Nevertheless, it turns out that universality holds; the local eigenvalue correlations of (at least) a large class of ensembles of Hermitian Wigner matrices converges to the same Wigner-Dyson distribution (\ref{eq:WD}). For an arbitrary ensembles $H$ of Hermitian Wigner matrices (as in Definition \ref{def}) whose entries decay sufficiently fast at infinity, in the sense that $\E \, |x_{jk}|^{K} , \E \, |x_{jj}|^{K} < \infty$, for a sufficiently large $K >0$, it was recently  proved in \cite{TV3} that \begin{equation}\label{eq:univ}
\frac{1}{\rho^k_{sc} (E)} p_H^{(k)} \left( E + \frac{x_1}{\rho_{sc} (E) N} , \dots , E + \frac{x_k}{\rho_{sc} (E) N} \right)  \to \det \, \left( \frac{\sin (\pi (x_j - x_\ell))}{(\pi (x_j - x_\ell))} \right)_{1 \leq j ,\ell \leq k} \end{equation} for any fixed $k \in \bN$,
as $N \to \infty$. Convergence here holds pointwise in $|E| < 2$ (actually, uniformly in $E\in [-2+\kappa ; 2 -\kappa]$, for any fixed $\kappa >0$), after integrating against a continuous and compactly supported observable $O (x_1, \dots , x_k)$. This result was obtained by extending the methods of \cite{ERSTVY}, where (\ref{eq:univ}) was already shown under the additional assumptions that $\E \, e^{|x_{ij}|^\alpha} < \infty$, $\E \, e^{|x_{jj}|^{\alpha}} < \infty$ and $\E \, x_{ij}^3 = 0$ (without this last assumption, (\ref{eq:univ}) was proven in \cite{ERSTVY} after averaging $E$ over an arbitrarily small, but fixed, interval). The correlation function $p_H^{(k)}$ is defined (similarly to $p_{\text{GUE}}$) by
\[ p_H^{(k)} (\mu_1, \dots , \mu_k) = \int d\mu_{k+1} \dots d\mu_N \, p_H (\mu_1, \dots , \mu_N) \] where $p_H$ is the joint probability density function of the eigenvalues of $H$. Note that the techniques of \cite{TV3,ERSTVY} (which are based on the methods developed in \cite{EPRSY,TV}; see next paragraph) cannot be easily extended to ensembles of Wigner matrices with different, non-Hermitian, symmetries. Universality (after integration of $E$ over an arbitrarily small, fixed, interval) for ensembles of real symmetric and quaternion Hermitian Wigner matrices was established in \cite{ESY4} using a different approach (in this paper we will need the result of universality pointwise in $E$; this is why we focus our attention on Hermitian matrices). Finally we observe that universality (after integration of $E$ over a small interval) was recently extended to ensembles of generalized Wigner matrices; see \cite{EYY1,EYY2}.

\medskip

The results of \cite{ERSTVY} were obtained by combining the methods proposed first in \cite{EPRSY} and then in \cite{TV}. In \cite{EPRSY}, universality was proven for Wigner matrices whose entries have a sufficiently regular law (and decay sufficiently fast at infinity). The first ingredient in \cite{EPRSY} was a proof of universality for matrices of the form $H = H_0 + s(N) V$, where $H_0$ is an arbitrary Hermitian Wigner matrix, $V$ is a GUE matrix, independent of $H_0$, and $s(N) \simeq N^{-1/2+\eps}$ measures the size of the Gaussian perturbation. Note that universality for matrices of the form $H=H_0 + s V$ was already proven in \cite{J} (whose result was then further improved in \cite{BP}), but only for fixed, $N$ independent, $s >0$. The second ingredient in \cite{EPRSY} was a time-reversal argument to compare the local correlations of the given Wigner matrix with those of a perturbed matrix of the form $H_0 + s(N) V$. In \cite{TV}, on the other hand, universality was proven for Hermitian Wigner matrices $H$, whose entries decay subexponentially fast at infinity, are supported on at least three points, and are such that $\E x_{jk}^3 = 0$. The main tool developed in \cite{TV} to show universality was a four-moment theorem comparing the local correlations of two ensembles whose entries have four matching moments.

\medskip

Both proofs, the one of \cite{EPRSY} and the one of \cite{TV}, relied on the convergence  to the semicircle law for the density of states on microscopic intervals. Eq. (\ref{eq:Wsc}) establishes the convergence of the density of states to the semicircle law on intervals whose size is independent of $N$, hence on intervals containing typically order $N$ eigenvalues ({\it macroscopic} intervals). What happens then on smaller intervals, whose size converge to zero as $N\to \infty$? This question is addressed in \cite{ESY1,ESY2,ESY3}. It follows from these works that the density of states converges to the semicircle law on arbitrary intervals, containing typically a large number of eigenvalues (the higher the number of eigenvalues, the smaller the fluctuations around the semicircle law). This conclusion is reached by comparing the Stieltjes transform of a Wigner matrix $H$ with the Stieltjes transform of the semicircle law.
The Stieltjes transform of the $N \times N$ Hermitian matrix $H$ is defined as the function (of $z\in \bC \backslash \bR$)
\[ m_N (z) = \frac{1}{N} \tr \, \frac{1}{H-z} =  \frac{1}{N} \sum_{\alpha=1}^N \frac{1}{\mu_\alpha - z} \, , \]
where $(\mu_1, \dots ,\mu_N)$ are the eigenvalues of $H$. The Stieltjes transform of the semicircle law, on the other hand, is defined by
\[ m_{sc} (z) = \int ds \frac{\rho_{sc} (s)}{s-z} = -\frac{z}{2} + \sqrt{\frac{z^2}{4}-1}\,. \]
{F}rom the convergence of $\text{Im } m_N (z)$ to $\text{Im } m_{sc} (z)$, one can deduce convergence of the density of states on intervals of size comparable with ${\text Im } \, z$. The advantage of working with the Stieltjes transform, instead of directly with the density of states, is the fact that $m_{sc} (z)$ satisfies a fixed point equation which is stable when ${\text Re } z$ is away from the spectral edges $\pm 2$. Using this fixed point equation (and an upper bound for the density of states on microscopic intervals), it is shown in \cite{ESY3} that,
if the entries of the matrix $H$ have sub-Gaussian tails (in the sense that $\E \, e^{\alpha x_{jk}^2} < \infty$, for some $\alpha >0$), and if $|E| < 2$, there exist constants $C,c >0$ such that
\begin{equation}\label{eq:mic-sc-ST}
\P \left( \left| m_N \left(E+ i \frac{K}{N} \right)- m_{sc} \left(E+i\frac{K}{N}\right) \right| \geq \delta \right) \leq C e^{-c \delta \sqrt{K}}
\end{equation}
for every $\delta >0$ small enough, and every $N$ sufficiently large. Note that this result was recently improved in \cite{EYY3}, where the convergence of the Stieltjes transform $m_N (E+i\eta)$ is shown to hold uniformly in $E$, with an error of size $(N\eta)^{-1}$. {F}rom (\ref{eq:mic-sc-ST}), it follows that, for every $\delta >0$, and $|E| <2$,
\begin{equation}\label{eq:mic-sc} \lim_{K \to \infty} \lim_{N \to \infty} \P \left( \left| \frac{\cN \left[E - \frac{K}{2N}, E + \frac{K}{2N} \right]}{K} - \rho_{sc} (E) \right| \geq \delta \right) = 0\,. \end{equation}
This result, shown in \cite{ESY3}, establishes the convergence to the semicircle law on intervals typically containing a number of eigenvalues of order one, independent of $N$ ({\it microscopic} intervals). {F}rom (\ref{eq:mic-sc}), it is also possible to obtain convergence to the semicircle law on intermediate scales. If $\eta (N) >0$ is such that $\eta (N) \to 0$ and $N\eta (N) \to \infty$ as $N \to \infty$, then
\begin{equation}\label{eq:etaN} \lim_{N \to \infty} \P \left( \left| \frac{\cN \left[E- \frac{\eta (N)}{2}; E + \frac{\eta (N)}{2} \right]}{ N \eta (N)} - \rho_{sc} (E) \right| \geq \delta \right) = 0 \, . \end{equation}
If $\eta (N) \lesssim 1/N$, so that $N\eta (N)$ does not tend to infinity as $N \to \infty$, then the typical number of eigenvalues in the interval of size $\eta (N)$ around $E$ remains bounded (it converges to zero, if $\eta (N) \ll 1/N$), and therefore the fluctuations of the density of states are certainly important. In this sense, the result (\ref{eq:mic-sc}) is on the optimal scale, and we cannot expect it to hold for smaller intervals.

\medskip

Consider now the average density of states on an interval of size $\eta (N) >0$ around $E$, defined as
\[ \E \, \frac{\cN \left[ E-\frac{\eta (N)}{2} ; E + \frac{\eta (N)}{2} \right]}{N\eta (N)}\,. \]
If $\eta (N)$ is such that $N \eta (N) \to \infty$ as $N \to \infty$, it follows from (a quantitative version of) Eq. (\ref{eq:etaN}) that we also have convergence of the average density of states to the semicircle law:
\begin{equation}\label{eq:ados0} \lim_{N \to \infty}  \E \, \frac{\cN \left[ E-\frac{\eta (N)}{2} ; E + \frac{\eta (N)}{2} \right]}{N\eta (N)} = \rho_{sc} (E) \, . \end{equation}
If $\eta (N) \lesssim (1/N)$, we do not have convergence to the semicircle law in probability ((\ref{eq:etaN}) is not true, in this case), but we may still ask whether the average density of states converges.

\medskip

A first important observation to answer this question is the fact that, if the probability law of the entries is sufficiently regular, the average density of states remains bounded on arbitrarily small scales. More precisely, under the assumption
\begin{equation}\label{eq:ass-up} \int \left| \frac{h' (s)}{h(s)} \right|^4 \, h(s) ds < \infty \end{equation}
it is proven in \cite{MS} (extending a previous result from \cite{ESY3}) that, for every $\kappa >0$, there exists a constant $C = C(\kappa) >0$ with
\begin{equation}\label{eq:up}
\E \, \cN \left[E-\frac{\eta}{2}; E+ \frac{\eta}{2} \right] \leq \E \, \cN^2 \left[E-\frac{\eta}{2}; E+ \frac{\eta}{2} \right] \leq  C N \eta \end{equation}
for all $\eta >0$, all $N \geq 10$, and all $E \in [-2+\kappa, 2-\kappa]$. Note that the upper bound on the expectation of the density of states also implies an upper bound on the expectation of the imaginary part of the Stieltjes transform. In fact, using (\ref{eq:up}) and a dyadic decomposition, we obtain
\begin{equation} \label{eq:up-m}\begin{split}
\E \, \text{Im } m_N (E+i\eps) \leq \; & \frac{1}{N} \E \sum_{\alpha} \frac{\eps}{(\mu_\alpha -E)^2 + \eps^2} \\  \leq \; & \E \, \frac{\cN \left[E- \eps ; E+\eps \right]}{N\eps} + \frac{\eps}{N} \sum_{\ell \geq 0} \E \, \frac{\cN \left[ E-2^{\ell+1} \eps ; E-2^{\ell} \eps \right] \cup \left[E+2^{\ell} \eps ; E+2^{\ell+1} \eps \right]}{2^{2\ell} \eps^2} \\ \lesssim \; &  1 + \sum_{\ell \geq 0} \frac{1}{2^{\ell}} \lesssim 1 \, .
\end{split} \end{equation}

\medskip

Another important remark concerning the average density of states on small intervals follows from universality. Consider an ensemble of hermitian Wigner matrices such that (\ref{eq:up}) holds true and a family of intervals of size $\eta (N) = \eps /N$, for a fixed, $N$ independent, $\eps >0$. Then we have, from (\ref{eq:univ}) with $k=1$,
\begin{equation} \label{eq:uni-0}
\E \frac{\cN \left[ E-\frac{\eps}{2N} ; E + \frac{\eps}{2N} \right]}{\eps} = \int dx \, \frac{{\bf 1} (|x| \leq \eps/2)}{\eps} \, p^{(1)}_H (E+ \frac{x}{N}) \to \int dx \frac{{\bf 1} (|x| \leq \eps)}{\eps} \, \rho_{sc} (E) = \rho_{sc} (E) \, , \end{equation} as $N \to \infty$. Note that ${\bf 1} (|x| \leq \eps/2)$ is not continuous and therefore (\ref{eq:univ}) cannot be applied directly. However, using the upper bound (\ref{eq:up}), it is simple to approximate ${\bf 1} (|x| \leq \eps/2)$ by continuous functions and conclude (\ref{eq:uni-0}). For future reference, we observe that the convergence in (\ref{eq:uni-0}) holds uniformly in $E$ away from the spectral edges. More precisely, for every fixed $\kappa,\eps >0$, we have
\begin{equation}\label{eq:uni}
\lim_{N\to \infty} \sup_{|E| \leq 2-\kappa} \left| \E \frac{\cN \left[ E-\frac{\eps}{2N} ; E + \frac{\eps}{2N} \right]}{\eps} - \rho_{sc} (E) \right| = 0\,.
\end{equation}
This follows because the arguments of \cite{EPRSY,TV,ERSTVY} are clearly uniform in $E$, as long as $E$ stays away from the edges (in Proposition 3.3 of \cite{EPRSY} this uniformity is explicitly stated).

\medskip

Hence, universality implies that the average density of states on intervals of size $\eps /N$ still converges to the semicircle law, for any fixed, $N$ independent, $\eps >0$. What happens now on even smaller scales $\eta (N) \ll 1/N$? The main result of the present paper is a proof of the convergence of the average density of states to the semicircle law on arbitrarily small scales, under some regularity assumption on the law of the entries of $H$. First, in the next theorem we establish convergence of the expectation of the imaginary part of the Stieltjes transform $m_N (E+i\eta)$ to the imaginary part of $m_{sc} (E+i\eta)$ uniformly in $\eta>0$, as $N \to \infty$.
\begin{theorem}\label{thm:main}
Let $H$ be an ensemble of Hermitian Wigner matrices as in Definition \ref{def}, so that $\E \, e^{\nu x_{ij}^2} < \infty$ for some $\nu >0$. Suppose that the real and imaginary part of the off-diagonal entries have a common probability density function $h$ such that
\begin{equation}\label{eq:ass}
\int \, \left| \frac{h' (s)}{h(s)} \right|^6 \, h(s) ds < \infty, \qquad \text{and } \quad \int \left| \frac{h'' (s)}{h(s)} \right|^2 \, h(s) ds < \infty \, . \end{equation}
Then we have, for every $|E| <2$,
\[  \lim_{N \to \infty} \liminf_{\eta \to 0} \, \E \, \text{Im } m_N (E+i\eta) = \lim_{N\to \infty} \limsup_{\eta \to 0} \E \,  \text{Im } m_N (E+i\eta) =  \text{Im } m_{sc} (E) = \pi \rho_{sc} (E) \, . \]
The convergence is uniform in $E$, away from the spectral edges; for any $\kappa >0$,
\begin{equation}\label{eq:uniclaim}
\begin{split}
\lim_{N\to \infty} \sup_{|E| \leq 2-\kappa} \left| \liminf_{\eta \to 0} \E \, \text{Im } m_N (E+i\eta) - \rho_{sc} (E) \right| & = 0\, , \\
\lim_{N\to \infty} \sup_{|E| \leq 2-\kappa} \left| \limsup_{\eta \to 0} \E \, \text{Im } m_N (E+i\eta) - \rho_{sc} (E) \right| & = 0 \, .
\end{split}
\end{equation}
\end{theorem}


\medskip

{F}rom Theorem \ref{thm:main} we obtain in the next corollary the convergence of the average density of states to the semicircle law on arbitrarily small scales.
\begin{corollary}\label{cor}
Under the same assumptions as in Theorem \ref{thm:main}, and for any fixed $\kappa >0$, we have
\begin{equation}\label{eq:cor1} \begin{split}  \lim_{N \to \infty} \sup_{|E| \leq 2-\kappa} \left|  \liminf_{\eta \to 0} \, \E \, \frac{\cN\left[E-\frac{\eta}{2} ; E+ \frac{\eta}{2} \right]}{N\eta} - \rho_{sc} (E) \right|  &= 0\\
 \lim_{N \to \infty} \sup_{|E| \leq 2-\kappa} \left|  \limsup_{\eta \to 0} \, \E \, \frac{\cN\left[E-\frac{\eta}{2} ; E+ \frac{\eta}{2} \right]}{N\eta} - \rho_{sc} (E) \right|  &= 0 \end{split} \end{equation}
Moreover, for every sequence $\eta (N) >0$ with $\eta (N) \to 0$ as $N \to \infty$, we find
\begin{equation}\label{eq:cor2} \lim_{N \to \infty} \sup_{|E| < 2-\kappa} \left| \E  \, \frac{\cN \left[E-\frac{\eta (N)}{2}; E+\frac{\eta (N)}{2}\right]}{N\eta (N)}  - \rho_{sc} (E) \right| = 0 \,. \end{equation}
\end{corollary}

We expect similar results to hold also for ensembles of Wigner matrices with different symmetries (real symmetric and quaternion Hermitian). The main tool that we use to show Theorem \ref{thm:main}, namely Proposition \ref{prop}, can be easily extended to ensembles with different symmetries. However, to conclude the proof of Theorem \ref{thm:main}, we also need the universality result (\ref{eq:univ}) (for $k=1$ only) to hold {\it pointwise} in $E$ (this is used in (\ref{eq:uni-0}), (\ref{eq:uni})). So far, pointwise in $E$ universality for real symmetric and quaternion Hermitian ensembles is only known, from \cite{TV}, under the assumption that the first four moment of the entries match exactly the corresponding Gaussian moments. Thus, our theorem extends, so far, only to these special examples of real symmetric and quaternion Hermitian Wigner ensembles.

\medskip

Observe that while the convergence in (\ref{eq:mic-sc-ST}) and (\ref{eq:mic-sc}) is a result on the scale $\eta (N) = K/N$ for a large but fixed $K >0$, and universality is a result about correlations on the scale $1/N$, Theorem \ref{thm:main} and Corollary \ref{cor} deal with the density of states on arbitrarily small scales. Understanding the limit $N \to \infty$ of the average density of states, {\it uniformly} in the size $\eta$ of the interval, is the main challenge in showing Theorem \ref{thm:main} and Corollary \ref{cor}.

\medskip

We start by proving that Corollary \ref{cor} follows from Theorem \ref{thm:main} (here we use the upper bound (\ref{eq:up}), and the fact that (\ref{eq:ass}) implies (\ref{eq:ass-up})).
\begin{proof}[Proof of Corollary \ref{cor}]
Let $\eta >0$ and $\eps < \eta^2$. Then, we consider, for arbitrary $E\in (-2,2)$,
\[ \begin{split}
I := & \frac{1}{\eta} \E \, \int_{E-\eta /2}^{E+\eta /2} d \wt{E}  \, \text{Im } m_N (\wt{E}+i\eps) \\
=& \frac{1}{N\eta} \E \sum_{\alpha=1}^N \int_{E-\eta /2}^{E+\eta /2} d \wt{E}  \, \frac{\eps}{(\mu_\alpha-\wt{E})^2 + \eps^2} \\
=& \frac{1}{N\eta} \E \, \sum_{\alpha=1}^N \left[ \text{arctg } \left( \frac{\mu_\alpha- \left(E- \frac{\eta}{2}\right)}{\eps} \right) - \text{arctg } \left( \frac{\mu_\alpha- \left(E + \frac{\eta}{2}\right)}{\eps} \right) \right]
\end{split} \]
Now, we observe that there exists a universal constant $C>0$ such that
\[ \frac{\pi}{2} - \frac{1}{x} \leq \text{arctg } x \leq \frac{\pi}{2} \]
for all $x > C$, and such that
\[ -\frac{\pi}{2} \leq \text{arctg } x \leq -\frac{\pi}{2} - \frac{1}{x} \]
for all $x < -C$. Therefore, we obtain (for all $\eps$ sufficiently small),
\begin{equation}\label{eq:up-I} \begin{split}
I \leq \; & \pi \E \, \frac{\cN \left[ E-\frac{\eta}{2} - \sqrt{\eps} ; E + \frac{\eta}{2}+\sqrt{\eps} \right]}{N\eta} +  \frac{1}{N\eta} \E \, \sum_{\{\alpha: \mu_\alpha < E -\frac{\eta}{2}-\sqrt{\eps}\}} \frac{\eps}{(E - \eta/2) - \mu_\alpha} \\ &+ \frac{1}{N\eta} \E \, \sum_{\{ \alpha: \mu_\alpha > E + \frac{\eta}{2} +\sqrt{\eps} \}} \frac{\eps}{\mu_\alpha - (E + \eta/2)} \\ \leq \; & \pi \E \, \frac{\cN \left[ E-\frac{\eta}{2} - \sqrt{\eps} ; E + \frac{\eta}{2}+\sqrt{\eps} \right]}{N\eta} + \frac{\sqrt{\eps}}{\eta} \\ \leq \; & \pi  \E \, \frac{\cN \left[ E-\frac{\eta}{2}; E + \frac{\eta}{2} \right]}{N\eta} + C \frac{\sqrt{\eps}}{\eta}
\end{split}
\end{equation}
where we used the upper bound (\ref{eq:up}). Analogously, we can show the lower bound
\begin{equation} \label{eq:low-I} I \geq  \pi  \E \, \frac{\cN \left[ E-\frac{\eta}{2}; E + \frac{\eta}{2} \right]}{N\eta} - C
\frac{\sqrt{\eps}}{\eta} \end{equation}
This implies that
\begin{equation}\label{eq:bd2} \begin{split}
\pi  \E \, \frac{\cN \left[ E-\frac{\eta}{2}; E + \frac{\eta}{2} \right]}{N\eta} = &\; \liminf_{\eps \to 0} \frac{1}{\eta} \, \int_{E-\eta /2}^{E+\eta /2} d \wt{E}  \, \E \, \text{Im } m_N (\wt{E}+i\eps) \\ = & \; \frac{1}{\eta} \int_{E-\eta/2}^{E+\eta/2} d \wt{E} \, \liminf_{\eps \to 0}
\E \, \text{Im } m_N (\wt{E}+i\eps)
\\ = & \; m_{sc} (E) \\ &+\frac{1}{\eta} \int_{E-\eta/2}^{E+\eta/2} d \wt{E} \, \left( m_{sc} (\wt{E}) - m_{sc} (E) \right) \\
& +  \frac{1}{\eta} \int_{E-\eta/2}^{E+\eta/2} d \wt{E} \, \left( \liminf_{\eps \to 0} \E \, \text{Im } m_N (\wt{E}+i\eps) - m_{sc} (\wt{E}) \right)
\end{split}
\end{equation}
where, in the second line, we used the dominated convergence theorem (and the upper bound (\ref{eq:up-m})). Since $m_{sc} (E) = \pi \rho_{sc} (E)$, we obtain
\[ \begin{split}
\left| \E \, \frac{\cN \left[ E-\frac{\eta}{2}; E + \frac{\eta}{2} \right]}{N\eta} - \rho_{sc} (E) \right| \leq \frac{1}{\pi} \sup_{|E| \leq 2- \kappa} \left| \liminf_{\eps \to 0} \E \, \text{Im } m_N (\wt{E}+i\eps) - m_{sc} (\wt{E}) \right| + \frac{\eta}{4\pi \sqrt{\kappa}}
\end{split} \]
for all $|E| < 2-\kappa-\eta/2$ (to estimate the second term on the r.h.s. of (\ref{eq:bd2}),  we used the bound $m'_{sc} (E) \leq \kappa^{-1/2}$ valid for all $|E| \leq 2 - \kappa$). Letting $\eta \to 0$, we conclude that
\[ \begin{split}
&\left| \liminf_{\eta \to 0} \E \, \frac{\cN \left[ E-\frac{\eta}{2}; E + \frac{\eta}{2} \right]}{N\eta} - \rho_{sc} (E) \right| \leq \pi^{-1} \sup_{|E| \leq 2- \kappa} \left| \liminf_{\eps \to 0} \E \, \text{Im } m_N (\wt{E}+i\eps) - m_{sc} (\wt{E}) \right| \qquad \text{and}  \\
&\left| \limsup_{\eta \to 0} \E \, \frac{\cN \left[ E-\frac{\eta}{2}; E + \frac{\eta}{2} \right]}{N\eta} - \rho_{sc} (E) \right| \leq \pi^{-1}  \sup_{|E| \leq 2- \kappa} \left| \liminf_{\eps \to 0} \E \, \text{Im } m_N (\wt{E}+i\eps) - m_{sc} (\wt{E}) \right| \,.
\end{split} \]
for all $|E| < 2-\kappa$. Eq. (\ref{eq:cor1}) follows now, taking the limit $N \to \infty$, from
(\ref{eq:uniclaim}).  Eq. (\ref{eq:cor2}) can be proven similarly.
\end{proof}

The proof of Theorem \ref{thm:main} is based on the following crucial proposition.
\begin{proposition} \label{prop}
Let $H$ be an ensemble of Hermitian Wigner matrices as in Definition \ref{def}, so that $\E \, e^{\nu |x_{ij}|^2} < \infty$ for some $\nu >0$. Suppose that the real and imaginary part of the off-diagonal entries have a common probability density function $h$ such that
\begin{equation}\label{eq:ass2}
\int \, \left| \frac{h' (s)}{h(s)} \right|^6 \, h(s) ds < \infty, \qquad \text{and } \quad \int \left| \frac{h'' (s)}{h(s)} \right|^2 \, h(s) ds < \infty \, .
\end{equation}
Fix $\kappa >0$. Then there exists a constant $C >0$ such that
\begin{equation}\label{eq:prop}
\left| \frac{\rd}{\rd E} \, \E \, \text{Im } m_N (E + i\eta) \right| \leq C N
\end{equation}
holds for all $E \in (-2+\kappa,2-\kappa)$, for all $0< \eta \leq 1/N$, for all $N \in \bN$ large enough.
\end{proposition}

Note that Proposition \ref{prop}, whose proof is deferred to Section \ref{sec:prop}, can be easily extended to ensembles of Wigner matrices with different symmetry (real symmetric or quaternion hermitian ensembles). Next, we show how the statement of Theorem \ref{thm:main} follows from Proposition \ref{prop}.

\begin{proof}[Proof of Theorem \ref{thm:main}]
We start by observing that, for any $\delta >0$,
\[ \begin{split}
I := & \; \frac{N}{\delta} \int_{E-\frac{\delta}{2N}}^{E+\frac{\delta}{2N}} d\wt{E} \, \E \, \text{Im } m_N (\wt{E} + i \eta) \\ = \; & \E \, \text{Im } m_N (E+i\eta) + \frac{N}{\delta} \int_{E-\frac{\delta}{2N}}^{E+\frac{\delta}{2N}} d\wt{E} \, \left(\E \, \text{Im } m_N (\wt{E} + i \eta) - \E \, \text{Im } m_N (E+i\eta) \right) \\ = \; &  \E \, \text{Im }  m_N (E+i\eta) + \frac{N}{\delta} \int_{E-\frac{\delta}{2N}}^{E+\frac{\delta}{2N}} d\wt{E} \, \int_{\wt{E}}^E  ds \, \frac{d}{ds} \, \E \, \text{Im } m_N (s + i \eta)
\end{split}\]
Therefore, from Proposition \ref{prop}, we find
\begin{equation}\label{eq:I-m_N} \begin{split}
\left| \E \, \text{Im } m_N (E+i\eta) - I \right| \leq & \; \frac{CN^2}{\delta} \int_{E-\frac{\delta}{2N}}^{E+\frac{\delta}{2N}} d\wt{E} \, |E-\wt{E}|  \leq C \delta
\end{split} \end{equation}
for all $E \in [-2+\kappa ; 2-\kappa]$, all $0<\eta <1/N$, and all $N$ sufficiently large.
Next we observe that
\[ \begin{split}
I = & \; \frac{1}{\delta} \E \, \sum_{\alpha} \int_{E-\frac{\delta}{2N}}^{E+\frac{\delta}{2N}} d\wt{E} \, \, \frac{\eta}{(\mu_\alpha - \wt{E})^2 + \eta^2} \\
=& \; \frac{1}{\delta} \E \sum_{\alpha} \left[ \text{arctg } \left( \frac{\mu_\alpha- \left(E- \frac{\delta}{2N}\right)}{\eta} \right) - \text{arctg } \left( \frac{\mu_\alpha- \left(E + \frac{\delta}{2N}\right)}{\eta} \right) \right] \\
= & \; \pi \, \E \frac{\cN \left[ E- \frac{\delta}{2N} ; E + \frac{\delta}{2N} \right]}{\delta} + O \left(\frac{\sqrt{\eta} N}{\delta} \right)
\end{split} \]
where we proceeded as in (\ref{eq:up-I}), (\ref{eq:low-I}). Last equation, together with (\ref{eq:I-m_N}), implies that
\[ \begin{split} &\left| \liminf_{\eta \to 0} \E \, \text{Im } m_N (E+i\eta) -  \pi \, \E \frac{\cN \left[ E- \frac{\delta}{2N} ; E + \frac{\delta}{2N} \right]}{\delta} \right| \leq C \delta \qquad \text{and} \\
& \left| \limsup_{\eta \to 0} \E \, \text{Im } m_N (E+i\eta) -  \pi \, \E \frac{\cN \left[ E- \frac{\delta}{2N} ; E + \frac{\delta}{2N} \right]}{\delta} \right| \leq C \delta  \end{split} \]
where the constant $C>0$ is independent of $E$, for $E \in (-2+\kappa;2-\kappa)$ and of $N$, for all $N$ large enough. It follows from (\ref{eq:uni}) that
\[ \lim_{N \to \infty} \sup_{E \in [-2+\kappa;2-\kappa]} \left| \E \frac{\cN \left[ E- \frac{\delta}{2N} ; E + \frac{\delta}{2N} \right]}{\delta} - \rho_{sc} (E) \right| = 0 \]
for every fixed $\delta >0$. Note that (\ref{eq:uni}) follows from the universality result (\ref{eq:univ}), with $k=1$, obtained in \cite{ERSTVY} under the assumption $\E \, x_{ij}^3 = 0$. Therefore, we conclude that
\begin{equation}\label{eq:final}
\begin{split}
&\lim_{N \to \infty} \sup_{E \in [-2+\kappa; 2-\kappa]} \left|  \liminf_{\eta \to 0} \E \, \text{Im } m_N (E+i\eta) - m_{sc} (E)\right| \leq C \delta \qquad \text{and} \\
&\lim_{N \to \infty} \sup_{E \in [-2+\kappa; 2-\kappa]} \left|  \liminf_{\eta \to 0} \E \, \text{Im } m_N (E+i\eta) - m_{sc} (E)\right| \leq C \delta
\end{split}
\end{equation}
Since $\delta >0$ is arbitrary, last equation implies (\ref{eq:uniclaim}).
\end{proof}

\section{Proof of Proposition \ref{prop}}
\setcounter{equation}{0}
\label{sec:prop}

The goal of this section is to prove Proposition \ref{prop}. We are going to prove that there exists a universal constant $C>0$ such that \begin{equation}\label{eq:eps}
\left| \frac{\rd}{\rd E} \, \E \, \text{Im } m_N \left(E + i \frac{\eps}{N} \right) \right| \leq C N
\end{equation}
for all $E \in (-2+\kappa, 2-\kappa)$, $N \in \bN$ sufficiently large, and $0 < \eps \leq 1$.
To show (\ref{eq:eps}), we start by writing
\[ m_N \left(E + i \frac{\eps}{N} \right)  = \frac{1}{N} \sum_{j=1}^N \frac{1}{H-E-i\frac{\eps}{N}} (j,j) =\frac{1}{N} \sum_{j=1}^N \frac{1}{h_{jj}-E-i\frac{\eps}{N} - \frac{1}{N}\sum_{\alpha}\frac{\xi_{\alpha}^{(j)}}{\lambda^{(j)}_{\alpha}-E-i\frac{\eps}{N}}} \]
where $\xi_{\alpha}^{(j)} = N |\bu_{\alpha}^{(j)}\cdot \bba^{(j)}|^2$, and where $\lambda^{(j)}_{\alpha}$ and ${\bf u}^{(j)}_{\alpha}$ are the eigenvalues and the  eigenvectors of the $(N-1)\times (N-1)$ minor $B^{(j)}$ of $H$, obtained by removing the $j$-th row and the $j$-th column (we will assume that the $\lambda^{(j)}_\alpha$ are ordered, in the sense that, for every $j \in \{ 1, \dots , N \}$, $\lambda^{(j)}_1 \leq \lambda^{(j)}_2 \leq \dots \leq \lambda^{(j)}_{N-1}$). Taking the expectation, we find
\begin{equation}\label{eq:xii}
\E m_N \left( E+i \frac{\e}{N} \right) = \E \, \frac{1}{h - E - i\frac{\eps}{N} - \frac{1}{N}\sum_{\alpha}\frac{\xi_{\alpha}}{\lambda_{\alpha} - E - i\frac{\eps}{N}}}
\end{equation}
where we put $h = h_{11}$, $\xi_\alpha = N |\bba \cdot \bu_\alpha|^2$ , where $\bba = \bba^{(1)} = (h_{12}, \dots , h_{1N})$, and where $\bu_\alpha = \bu_\alpha^{(1)}$, $\lambda_\alpha = \lambda^{(1)}_\alpha$ are the eigenvectors and the eigenvalues of $B = B^{(1)}$ (by symmetry, the expectation of $(H-z)^{-1} (j,j)$ is independent of $j = 1, \dots ,N$). Taking the imaginary part, we obtain
\begin{equation}\label{eq:ImEm} \text{Im } \E \, m_N \left( E+i \frac{\e}{N} \right) = \E \,\frac{ \e/N + \sum_{\alpha}c_{\alpha}\xi_{\alpha}}{(h - E -\sum_{\alpha}d_{\alpha}\xi_{\alpha})^2 + (\e/N + \sum_{\alpha}c_{\alpha}\xi_{\alpha})^2} \end{equation}
where we defined
\begin{align}
c_{\alpha} &= \frac{\e}{N^2(\lambda_{\alpha} - E)^2 + \e^2}\\
d_{\alpha} &= \frac{N(\lambda_{\alpha}-E)}{N^2(\lambda_{\alpha}-E)^2+\e^2}\,.
\label{eq:dbeta}
\end{align}
We compute next the derivative of (\ref{eq:ImEm}) with respect to $E$ (note that $c_\alpha$ and $d_\alpha$ depend on $E$). We obtain
\begin{equation}\label{eq:der}
\begin{split}
\frac{d}{dE} \, \E \, & m_N \left( E+i \frac{\e}{N} \right) \\ = \; &\E \, \left( \sum_{\alpha} c_{\alpha}' \, \xi_{\alpha}\right) \,  \frac{(h - E-\sum_{\alpha}d_{\alpha}\xi_{\alpha})^2 - (\sum_{\alpha} c_{\alpha} \xi_{\alpha})^2}{\left[(h - E -\sum_{\alpha}d_{\alpha}\xi_{\alpha})^2 +(\eps/N+ \sum_{\alpha}c_{\alpha}\xi_{\alpha})^2\right]^2}\\
&+2 \E \,  \left(1+ \sum_{\alpha} d_{\alpha}' \xi_{\alpha}\right) \,
 \frac{\left( \eps/N + \sum_{\alpha}c_{\alpha}\xi_{\alpha}\right) \left(h - E -\sum_{\alpha}d_{\alpha}\xi_{\alpha} \right)}{\left[(h - E -\sum_{\alpha}d_{\alpha}\xi_{\alpha})^2 + (\eps/N + \sum_{\alpha}c_{\alpha}\xi_{\alpha})^2\right]^2}
\end{split}
\end{equation}
where we defined the derivatives of $c_\alpha$ and $d_\alpha$ with respect to $E$:
\begin{align*}
c_{\alpha}' &= \e N\frac{N(\lambda_{\alpha} - E)}{(N^2(\lambda_{\alpha} - E)^2 + \e^2)^2}\\
d_{\alpha}' &= N\frac{N^2(\lambda_{\alpha}- E)^2-\e^2}{(N^2(\lambda_{\alpha}- E)^2+\e^2)^2}
\end{align*}
We are going to estimate the absolute value of (\ref{eq:der}) by first taking the expectation over the component of $\bba = (h_{12}, \dots , h_{1N})$ keeping $h = h_{11}$ and the minor $B$ fixed. Only later, we will take expectation over $B,h$. In other words, we bound
\begin{equation}\label{eq:I+II} \left| \frac{d}{dE} \, \E \,  m_N \left( E+i \frac{\e}{N} \right) \right| \leq \E_{B,h} \left[ |\text{I}| + |\text{II}| + |\text{III}| \right] \end{equation}
where we set
\begin{equation}\label{eq:I}
\text{I } = \E_{\bba} \, \left( \sum_{\alpha} c_{\alpha}' \, \xi_{\alpha}\right) \,  \frac{(h - E-\sum_{\alpha}d_{\alpha}\xi_{\alpha})^2 - (\sum_{\alpha} c_{\alpha} \xi_{\alpha})^2}{\left[(h - E -\sum_{\alpha}d_{\alpha}\xi_{\alpha})^2 +(\eps/N+ \sum_{\alpha}c_{\alpha}\xi_{\alpha})^2\right]^2},
\end{equation}
\begin{equation}\label{eq:II}
\text{II } =2 \E_{\bba} \, \left(\sum_{\alpha} d_{\alpha}' \xi_{\alpha}\right)  \, \frac{\left( \eps / N + \sum_{\alpha}c_{\alpha}\xi_{\alpha}\right) \left(h - E -\sum_{\alpha}d_{\alpha}\xi_{\alpha} \right)}{\left[(h - E -\sum_{\alpha}d_{\alpha}\xi_{\alpha})^2 + (\eps/N + \sum_{\alpha}c_{\alpha}\xi_{\alpha})^2\right]^2}
\end{equation}
and
\begin{equation}\label{eq:III}
\text{III } =2 \, \E_{\bba} \,  \frac{\left( \eps/N + \sum_{\alpha}c_{\alpha}\xi_{\alpha}\right) \left(h - E -\sum_{\alpha}d_{\alpha}\xi_{\alpha} \right)}{\left[(h - E -\sum_{\alpha}d_{\alpha}\xi_{\alpha})^2 + (\eps/N + \sum_{\alpha}c_{\alpha}\xi_{\alpha})^2\right]^2}
\end{equation}

\medskip

In order to bound these contributions, we need to select indices of eigenvalues of $B$ playing an important role. We will need some of these eigenvalues to be at distances larger than $\eps/N$ from $E$ (to make sure that the corresponding coefficient $d_\alpha$ is not small). In order to define these indices, we need to exclude the (extremely unlikely) event that less than eight eigenvalues of $B$ are outside the interval $[E-\eps/N ; E+ \eps/N]$. We define therefore the ``good'' event
\begin{equation}\label{eq:Omega} \Omega = \left\{ \text{there exist at least eight eigenvalues of $B$ outside the interval } [E-\frac{\eps}{N} ; E+ \frac{\eps}{N}] \right\} \end{equation}
We will show now that, on the good event,
\begin{equation}\label{eq:claimO} \E_{B,h} \, {\bf 1} (\Omega) \left( |\text{I}| +  |\text{II}|+  |\text{III}| \right) \lesssim N \end{equation}
At the end of the proof, we will discuss the expectation on the ``bad'', complementary,  event $\Omega^c$. Until then, our analysis will always be restricted to the ``good'' set $\Omega$.

\medskip

In order to show (\ref{eq:claimO}), we choose, for a fixed realization of $B$, the index $\beta_0 \in \{ 1, \dots , N-1\}$ so that
\begin{equation}\label{eq:beta0} |\lambda_{\beta_0} - E| = \min_{\alpha = 1, \dots , N-1} |\lambda_\alpha -E| \, . \end{equation}
Moreover, on the set $\Omega$, we fix recursively the indices $\beta_j$, $j=1,\dots,8$, so that
\begin{equation}\label{eq:indices} |\lambda_{\beta_j} - E| = \min \{ |\lambda_{\gamma} - E| : |\lambda_{\gamma} - E| \geq \eps/N, \beta_j \not = \beta_0, \dots , \beta_{j-1} \} \, \end{equation} In other words, $\beta_1, \dots , \beta_8$ are the indices of the eight  eigenvalues of $B$ (different from $\beta_0$) which are closest to $E$ under the condition that they are not $\beta_0$, and that their distance to $E$ is at least $\eps/N$ (this conditions guarantees the monotonicity of the coefficients $|d_{\beta_j}|$). Let
\begin{equation}\label{eq:Delta} \Delta := N|\lambda_{\beta_8} -E| \end{equation}
Then, we have
\begin{equation}\label{eq:delta} \frac{1}{2\Delta} \leq |d_{\beta_8}| \leq \dots \leq |d_{\beta_1}| \leq \frac{1}{\eps} \end{equation}
and, similarly,
\begin{equation}\label{eq:cdelta}
\frac{\eps}{2\Delta^2} \leq |c_{\beta_8}| \leq \dots \leq |c_{\beta_1}| \leq \frac{1}{\eps}\end{equation}

\medskip

We start by controlling the term $\text{I}$, defined in (\ref{eq:I}). We have
\[ \begin{split}
|\text{I}| &\leq \E_{\bba} \frac{\sum_{\alpha} |c_{\alpha}'| \xi_{\alpha}}{(h - E -\sum_{\alpha}d_{\alpha}\xi_{\alpha})^2 + (\sum_{\alpha}c_{\alpha}\xi_{\alpha})^2} \leq  \sum_{\alpha} |c_{\alpha}'| \, \E_{\bba} \, \frac{ \xi_\alpha}{(h - E -\sum_{\alpha}d_{\alpha}\xi_{\alpha})^2 + (\sum_{\alpha}c_{\alpha}\xi_{\alpha})^2}
\end{split}
\]
because $c'_\alpha$ is independent of $\bba$ (the coefficients $c_\alpha$, $d_\alpha$ only depend on the eigenvalues $\lambda_\alpha$ of the minor $B$; therefore they are independent of the first row and column). Let $\bb = \sqrt{N} \bba = (b_1, \dots , b_{N-1})$. Then \[\begin{split} \E_{\bba} \, & \frac{ \xi_\alpha}{(h - E -\sum_{\alpha}d_{\alpha}\xi_{\alpha})^2 + (\sum_{\alpha}c_{\alpha}\xi_{\alpha})^2} \\ &= \int  d\bb \, d {\bar{\bb}} \prod_{j=1}^{N-1} h (\text{Re } b_j) h(\text{Im } b_j )  \, \frac{|\bb \cdot \bu_\alpha|^2}{(h - E -\sum_{\alpha} d_{\alpha} |\bb \cdot \bu_{\alpha}|^2)^2 + (\sum_{\alpha} c_{\alpha} |\bb \cdot \bu_{\alpha}|^2 )^2}
\end{split}
\]
We introduce new variables $z_\alpha = \bb \cdot \bu_\alpha$, for $\alpha = 1, \dots , N-1$ (recall that $\bu_\alpha$, for $\alpha = 1,\dots , N-1$ are the $(N-1)$ normalized eigenvectors of the minor $B$). Let $U$ be the $(N-1) \times (N-1)$ matrix with rows $\bu_1, \dots \bu_{N-1}$; then $U$ is a unitary matrix, and $\bz = (z_1, \dots , z_{N-1}) = U \bb$. Hence
\begin{equation}\label{eq:zalpha}\begin{split} \E_{\bba} \, & \frac{ \xi_\alpha}{(h - E -\sum_{\alpha}d_{\alpha}\xi_{\alpha})^2 + (\sum_{\alpha}c_{\alpha}\xi_{\alpha})^2} \\ &= \int d\mu (\bz) \, \frac{|z_\alpha|^2}{(h - E -\sum_{\alpha} d_{\alpha} |z_{\alpha}|^2)^2 + (\sum_{\alpha} c_{\alpha} |z_{\alpha}|^2 )^2}\,.
\end{split}
\end{equation}
where we defined \begin{equation}\label{eq:dmu} d\mu (\bz) = \prod_{j=1}^{N-1} h (\text{Re } (U^* \bz)_j) h(\text{Im } (U^* \bz)_j ) d\bz d\bar{\bz} \, . \end{equation}
It follows from Proposition \ref{prop:I}, that, on the set $\Omega$,
\begin{equation}\label{eq:EaI}
|\text{I}| \leq \sum_\alpha |c'_\alpha| \, \min \left( \frac{\Delta^3}{\eps} , \frac{\Delta}{c_\alpha}, \frac{\Delta^{7/8}}{c_\alpha \, |d_{\beta_0}|^{1/8}}\right) \,. \end{equation}

Similarly, on $\Omega$, the contribution $\text{II}$ defined in (\ref{eq:I}) is bounded, using (\ref{eq:propII}) in Proposition \ref{prop:II}, by
\begin{equation}\label{eq:EaII} \begin{split} |\text{II}| \leq \; & 2 \sum_\alpha |d'_\alpha| \left| \int d\mu (\bz) \, |z_{\alpha}|^2  \,   \frac{\left(\frac{\eps}{N} + \sum_\gamma c_\gamma |z_{\gamma}|^2\right)  \left( h - E -\sum_{\alpha}d_{\alpha}|z_\alpha|^2\right)}{\left[(h - E -\sum_{\alpha}d_{\alpha}|z_{\alpha}|^2)^2 + (\frac{\eps}{N} + \sum_{\alpha}c_{\alpha}|z_{\alpha}|^2)^2\right]^2} \right| \\ \lesssim & \; \sum_\alpha |d'_\alpha| \min \left(\frac{\Delta}{|d_\alpha|} ,  \, \frac{\Delta^{7/8}}{|d_\alpha | \, |d_{\beta_0}|^{1/8}},\frac{\Delta}{c_{\alpha}}, \Delta^2 \right) \end{split} \end{equation}
and the term $\text{III}$ defined in (\ref{eq:III}) can be estimated by
\begin{equation}\label{eq:EaIII}\begin{split}
|\text{III}| \leq \; &2\left|  \int d\mu (\bz) \,  \frac{ \left(\frac{\eps}{N} + \sum_\gamma c_\gamma |z_{\gamma}|^2\right)  \left(h - E -\sum_{\alpha}d_{\alpha}|z_\alpha|^2\right)}{\left[(h - E -\sum_{\alpha}d_{\alpha}|z_{\alpha}|^2)^2 + (\frac{\eps}{N} + \sum_{\alpha}c_{\alpha}|z_{\alpha}|^2)^2\right]^2} \right| \lesssim \Delta^2 \end{split} \end{equation}

\medskip

Next, we take expectation over the randomness in $B$ (the r.h.s. of (\ref{eq:EaI}), (\ref{eq:EaII}), (\ref{eq:EaIII}) are already independent of $h=h_{11}$). First of all, we note that, from (\ref{eq:EaIII}),
\begin{equation}\label{eq:EBIII} \E_{B} \, {\bf 1} (\Omega) \, | \text{III}| \lesssim \E_B \,  {\bf 1} (\Omega) \, \Delta^2 \lesssim 1 \end{equation}
by Lemma \ref{lm:Delta}, part (1). To control $\E_B {\bf 1} (\Omega) |\text{I}|$ we use (\ref{eq:EaI}). Depending on the index $\alpha$, we are going to use different bounds. We define the sets of indices $S_1 = \{ \alpha : N|\lambda_\alpha - E| \leq \eps \}$, $S_2 = \{ \alpha: \eps \leq  N|\lambda_\alpha - E|  \leq 1\}$, $S_3 = \{ \alpha : N |\lambda_\alpha - E| \geq 1\}$. Then, we have
\[\begin{split}
|\text{I}| \leq\; & \Delta \, \sum_{\alpha \in S_1} \frac{|c'_\alpha|}{c_\alpha} +\Delta \, {\bf 1} (N|\lambda_{\beta_0}-E| \leq \eps) \sum_{\alpha \in S_2} \frac{|c'_\alpha|}{c_\alpha} + \Delta^{7/8} \, \frac{{\bf 1} (N|\lambda_{\beta_0} -E| \geq \eps)}{ |d_{\beta_0}|^{1/8}} \sum_{\alpha \in S_2} \frac{|c'_\alpha|}{c_\alpha} + \Delta^3 \sum_{\alpha \in S_3} \frac{|c'_\alpha|}{\eps}\end{split} \]
Since
\[ \frac{|c'_\alpha|}{c_\alpha} = N \, \frac{N(\lambda_{\alpha} - E)}{N^2(\lambda_{\alpha} - E)^2 + \e^2} \leq N \left\{ \begin{array}{ll}  \eps^{-1} &\quad \text{if } \alpha \in S_1, \\
\frac{1}{N|\lambda_\alpha - E|} &\quad \text{if } \alpha \in S_2, S_3, \end{array} \right. \]
we conclude that
\begin{equation} \label{eq:I-parti} \begin{split} |\text{I}| \lesssim \; &N \Delta \, \frac{\cN_B \left[E-\frac{\eps}{N} ; E+\frac{\eps}{N} \right]}{\eps}
\\ &+ N \Delta \, {\bf 1} (N|\lambda_{\beta_0}-E| \leq \eps) \sum_{\ell = 1}^{|\log \eps|} \frac{\cN_B \left[E- 2^{\ell} \frac{\eps}{N} ; E-2^{\ell-1} \frac{\eps}{N} \right] + \cN_B \left[ E+2^{\ell-1} \frac{\eps}{N} ; E+ 2^\ell \frac{\eps}{N} \right]}{2^{\ell} \eps} \\ &+ N\Delta^{7/8} \sum_{k =1}^{|\log \eps|}  (2^k \eps)^{1/8} \, {\bf 1} (2^{k-1} \eps \leq N|\lambda_{\beta_0} - E| \leq 2^k \eps ) \\ & \hspace{3cm} \times \sum_{\ell \geq k}^{|\log \eps|} \frac{\cN_B  \left[E- 2^{\ell} \frac{\eps}{N} ; E-2^{\ell-1} \frac{\eps}{N} \right] + \cN_B \left[ E+2^{\ell-1} \frac{\eps}{N} ; E+ 2^\ell \frac{\eps}{N} \right]}{2^{\ell} \eps}  \\ &+ N\Delta^3 \sum_{\ell =1}^{\infty} \frac{\cN_B \left[E- 2^\ell \frac{1}{N}; E-2^{\ell-1} \frac{1}{N} \right] +\cN_B \left[E+ 2^{\ell-1} \frac{1}{N}; E+2^{\ell} \frac{1}{N} \right]}{2^{3\ell}} \end{split}\end{equation}
where $\cN_B [A]$ denotes the number of eigenvalues of the minor $B$ in the interval $A$, and $\log$ is in basis two. In the third line, we use the fact that $|d_{\beta_0}| \geq (2^{k+1} \eps)^{-1}$, if $2^{k-1} \eps \leq N|\lambda_{\beta_0} - E| \leq 2^k \eps$, $k \geq 1$. In the fourth line, we use that, by definition of the index $\beta_0$, there are no eigenvalues of $B$ at distances smaller than $2^{k-1} \eps / N$ from $E$, under the condition that $N|\lambda_{\beta_0} - E| \geq 2^{k-1} \eps$. Using Lemma \ref{lm:Delta} (parts (2) and (3)) and Schwarz inequality we find
\[ \begin{split} \E_B  \,{\bf 1} (\Omega) \, \Delta \, {\bf 1} &(N |\lambda_{\beta_0} - E| \leq \eps) \,  \cN_B \left[ E-2^{\ell}\frac{\eps}{N} ; E-2^{\ell-1}\frac{\eps}{N} \right] \\ &\leq \left[ \E \, {\bf 1} (N |\lambda_{\beta_0} - E| \leq \eps) \right]^{1/2} \, \left[ \E\, {\bf 1} (\Omega) \,  \Delta^2 \, \cN^2_B \left[ E-2^{\ell}\frac{\eps}{N} ; E-2^{\ell-1}\frac{\eps}{N} \right] \right]^{1/2}  \lesssim 2^{\ell/2} \eps \end{split}\]
and, similarly,
\[  \begin{split} \E_B \, {\bf 1} (\Omega) \, \Delta^{7/8} \, &{\bf 1} (2^{k-1} \eps \leq N |\lambda_{\beta_0} - E| \leq 2^k \eps) \, \cN_B \left[ E-2^{\ell}\frac{\eps}{N} ; E-2^{\ell-1}\frac{\eps}{N} \right] \\ &\leq \left[ \E \, {\bf 1} (N |\lambda_{\beta_0} - E| \leq 2^k \eps) \right]^{1/2} \, \left[ \E \, {\bf 1} (\Omega) \, \Delta^{7/4} \, \cN^2_B \left[ E-2^{\ell}\frac{\eps}{N} ; E-2^{\ell-1}\frac{\eps}{N} \right] \right]^{1/2} \\ & \lesssim 2^{k/2} \, 2^{\ell/2} \eps \end{split}  \]
Applying part (2) of Lemma \ref{lm:Delta} in the first term on the r.h.s. of (\ref{eq:I-parti}), and part (1) and part (4) of Lemma \ref{lm:Delta} (after a Schwarz inequality) in the fourth term on the r.h.s. of (\ref{eq:I-parti}), we find
\begin{equation}\label{eq:EBI} \begin{split}
\E_B \, {\bf 1} (\Omega) \, |\text{I}| \lesssim \; & N \left( 1 + \sum_{\ell=1}^{\infty} 2^{-\ell/2} +\eps^{1/8}  \sum_{k=1}^{|\log \eps|} 2^{5k/8} \sum^{\infty}_{\ell \geq k} 2^{-\ell/2} + \sum_{\ell=1}^{\infty} 2^{-2\ell} \right) \\ \lesssim \; & N \left( 1 + \eps^{1/8} \sum_{k=1}^{|\log \eps|} 2^{k/8} \right) \lesssim N
\end{split} \end{equation}
Finally, we control $\E_B \,  {\bf 1} (\Omega) \, |\text{II}|$. {F}rom (\ref{eq:EaII}), we obtain
\begin{equation} \label{eq:pa-II} \begin{split}
|\text{II}| \leq\; & \Delta \, \sum_{\alpha \in S_1} \frac{|d'_\alpha|}{c_\alpha} + {\bf 1} (N|\lambda_{\beta_0}-E| \leq \eps) \sum_{\alpha \in S_2} \frac{|d'_\alpha|}{d_\alpha} + \Delta^{7/8} \, \frac{{\bf 1} (N|\lambda_{\beta_0} -E| \geq \eps)}{ |d_{\beta_0}|^{1/8}} \sum_{\alpha \in S_2} \frac{|d'_\alpha|}{d_\alpha} + \Delta^2 \sum_{\alpha \in S_3} |d'_\alpha| \end{split} \end{equation}
Since
\[ \frac{|d'_\alpha|}{|d_\alpha|} \leq \; N \frac{N|\lambda_\alpha -E|}{N^2 (\lambda_\alpha -E)^2 +\eps^2} \leq N \, \frac{1}{N|\lambda_\alpha -E|} \] if
$\alpha \in S_2, S_3$, and since
\[ \frac{|d'_\alpha|}{c_{\alpha}} \leq \frac{N}{\eps}, \qquad \text{ if } \alpha \in S_1, \] we conclude that $|\text{II}|$ can be bounded very similarly to (\ref{eq:I-parti}); the only difference is the last term, where the denominator $2^{3\ell}$ must be replaced by $2^{2\ell}$ and where $\Delta^3$ is replaced by $\Delta^2$. These changes are not important and therefore we obtain, as in (\ref{eq:EBI}), that
\[ \E_B \, {\bf 1} (\Omega) \, |\text{II}| \lesssim N \]
Together with (\ref{eq:EBIII}) and (\ref{eq:EBI}), this completes the proof of (\ref{eq:claimO}).

\medskip

Finally, we briefly explain how to bound the expectations of $|\text{I}|, |\text{II}|, |\text{III}|$ in the ``bad'' set $\Omega^c$. On $\Omega^c$, if $N \geq 11$, there are at least $4$ eigenvalues $\lambda_{\beta_j}$, $j =1,2,3,4$, in the interval $[E-\frac{\eps}{N} ; E+\frac{\eps}{N} ]$. This implies that $c_{\beta_j} \geq \eps^{-1}/2$, for $j=1,2,3,4$.
Hence, on $\Omega^c$, we can bound the term (\ref{eq:I}) by
\[\begin{split}
|\text{I}| \leq \; & \sum_\alpha  |c'_\alpha| \, \E_{\bba} \frac{\xi_\alpha}{(\sum_{j=1}^4 c_{\beta_j} \xi_{\beta_j})^2} \\ \lesssim \; &\eps^2 \sum_\alpha |c'_\alpha| \, \int d\mu (\bz) \frac{|z_\alpha|^2}{(\sum_{j=1}^4 |z_{\beta_j}|^2 )^2} \\ \lesssim \; & \eps^2 \sum_\alpha |c'_\alpha| \, \left( \int d\mu (\bz) \, |z_\alpha|^6 \right)^{1/3}  \left( \int d\mu (\bz) \frac{1}{(\sum_{j=1}^4 |z_{\beta_j}|^2 )^3} \right)^{2/3}
\end{split} \]
{F}rom Lemma \ref{lm:CS} and Lemma \ref{lm:pow}, we conclude that, on $\Omega^c$,
\[ \begin{split} |\text{I}| \lesssim\; & N \eps^3 \sum_{\alpha} \frac{1}{(N^2 (\lambda_\alpha - E)^2 + \eps^2)^{3/2}} \\ \leq \; &N \cN_B [E-\frac{1}{N} ; E + \frac{1}{N}]  + N \eps^3 \sum_{\alpha: N|\lambda_\alpha - E| \geq 1} \frac{1}{N^2 (\lambda_\alpha - E)^2} \\
\leq \; &N \cN_B [E-\frac{1}{N} ; E + \frac{1}{N}] + N \eps^3 \sum_{\ell \geq 1} \frac{\cN_B [E-2^\ell \frac{1}{N} ; E -2^{\ell-1} \frac{1}{N}] + \cN_B [E+ 2^{\ell-1} \frac{1}{N} ; E +2^{\ell} \frac{1}{N}]}{2^{2\ell}} \end{split} \]
Taking the expectation, we find, using Lemma \ref{lm:Delta},
\begin{equation}\label{eq:I-Oc} \E_B\, {\bf 1} (\Omega^c) |\text{I}| \lesssim N \left(1 + \eps^3 \sum_{\ell \geq 1} 2^{-\ell} \right) \lesssim N \end{equation}
Similarly, we can bound the term (\ref{eq:II}), on the set $\Omega^c$, by
\[ \begin{split} | \text{II}| \leq &\;  \sum_\alpha |d'_\alpha| \, \E_{\bba} \frac{\xi_\alpha}{(\sum_{j=1}^4 c_{\beta_j} \xi_{\beta_j})^2} \\ \lesssim \; &\eps^2 \sum_\alpha |d'_\alpha| \lesssim N \eps^2 \sum_\alpha \frac{1}{N^2 (\lambda_\alpha -E)^2 + \eps^2} \\ \lesssim \; & N \cN_B [E-\frac{1}{N} ; E + \frac{1}{N}] + N \eps^2 \sum_{\ell \geq 1} \frac{\cN_B [E-2^\ell \frac{1}{N} ; E -2^{\ell-1} \frac{1}{N}] + \cN_B [E+ 2^{\ell-1} \frac{1}{N} ; E +2^{\ell} \frac{1}{N}]}{2^{2\ell}}
\end{split} \]
which implies that
\begin{equation}\label{eq:II-Oc} \E_B\, {\bf 1} (\Omega^c) |\text{II}| \lesssim N \end{equation}
The term (\ref{eq:III}) can be estimated, on $\Omega^c$, by
\[ |\text{III}| \leq \E_{\bba} \, \frac{1}{(\sum_{j=1}^4 c_{\beta_j} \xi_{\beta_j})} \leq \eps^2 \int d\mu (\bz) \frac{1}{\sum_{j=1}^4 |z_{\beta_j}|^2} \lesssim \eps^2 \, . \]
{F}rom the last equation, together with (\ref{eq:I-Oc}), (\ref{eq:II-Oc}), we conclude that
\[ \E_B \,{\bf 1} (\Omega^c) \left( |\text{I}| + |\text{II}| + |\text{III}| \right)  \lesssim N \]
Combined with (\ref{eq:claimO}), this completes the proof of Proposition \ref{prop}.
\qed

\medskip

The following lemma, which was used above to estimates quantities depending on the eigenvalues of the minor $B$, is a collection of results which follow essentially from \cite{ESY3}.

\begin{lemma}\label{lm:Delta}
Fix $\kappa >0$, $E \in (-2+\kappa; 2-\kappa)$. Let the event $\Omega$ be defined as in (\ref{eq:Omega}), the (random) index $\beta_0$ be defined as in (\ref{eq:beta0}), and  the random variable $\Delta$ be defined as in (\ref{eq:Delta}). Then we have
\begin{itemize}
\item[1)]  For every $n \geq 0$,
\begin{equation}\label{eq:Delta0} \E \, {\bf 1} (\Omega) \, \Delta^n \leq 1 \end{equation}
\item[2)] For every $n \geq 0$,
\begin{equation}\label{eq:Delta2} \E \, {\bf 1} (\Omega) \, \Delta^n \, \cN^2_B \left[ E-\frac{\delta}{N} ; E+ \frac{\delta}{N} \right] \lesssim \delta \end{equation}  for every $0 < \delta <1$, for every $N \geq 10$.
\item[3)] We have the estimate
\begin{equation}\label{eq:Delta1} \E \, {\bf 1} (N|\lambda_{\beta_0} - E| \leq \delta) \lesssim \delta \end{equation} for every $\delta >0$.
\end{itemize}
\end{lemma}

\begin{proof}
Eq. (\ref{eq:Delta0}) follows from Theorem 3.3 in \cite{ESY3} (see also the discussion below (8.4) of \cite{ESY3}). The bound (\ref{eq:Delta2}) is proven essentially in Corollary 8.1 of \cite{ESY3}; instead of Theorem 3.4 in \cite{ESY3} we use Theorem 3.1 of \cite{MS} which holds true under the assumption (\ref{eq:ass2}). Observe that Corollary 8.1 in \cite{ESY3} is stated actually for the quantity ${\bf 1} (\cN_B \geq 1)$; however, in the proof, one uses ${\bf 1} (\cN_B \geq 1) \leq \cN^2_B$, and then one gives effectively a bound for the quantity in (\ref{eq:Delta2})). Eq. (\ref{eq:Delta1}) is a consequence of
\begin{equation}  \E \, {\bf 1} (N|\lambda_{\beta_0} - E| \leq \delta)  = \P ( N|\lambda_{\beta_0} - E| \leq \delta) \leq \P \left( \cN_B \, \left[ E-\frac{\delta}{N} ; E+ \frac{\delta}{N} \right]  \geq 1 \right) \lesssim \delta \end{equation}
by Theorem 3.4 of \cite{ESY3}.
\end{proof}

\section{Expectations over the row $\bba = (h_{12}, \dots , h_{1N})$}
\setcounter{equation}{0}
\label{sec:Ea}

In this section we prove two propositions, which are used in the proof of Proposition \ref{prop} to estimate the expectation over the row $\bba = (h_{12}, \dots , h_{1N})$ in terms of quantities depending on the eigenvalues of the minor $B$ (obtained from $H$ removing the first row and the first column). We will use the measure $d\mu (\bz)$ defined in (\ref{eq:dmu}), the indices $\beta_j$, $j=0,\dots,8$ defined in (\ref{eq:beta0}), (\ref{eq:indices}), and the length $\Delta$ defined in (\ref{eq:Delta}). The next proposition is used in the analysis of the term (\ref{eq:I}).

\begin{proposition}\label{prop:I}
Let $H$ be an ensemble of Hermitian Wigner matrices as in Definition \ref{def}, so that $\E \, e^{\nu |x_{ij}|^2} < \infty$ for some $\nu >0$. Let real and imaginary part of the off-diagonal entries have a common probability density function $h$ such that (\ref{eq:ass2}) holds true. Let $B$ be the $(N-1) \times (N-1)$ minor of $H$ obtained by removing the first row and the first column of $H$. Suppose the randomness in $B$ is such that  the event $\Omega$, defined in (\ref{eq:Omega}), is satisfied. Let the measure $d\mu (\bz)$ be defined as in (\ref{eq:dmu}). Then, for every $\alpha = 1, \dots , N-1$, we have
\begin{equation}\label{eq:propI} \begin{split}
\text{A} := \int  d\mu (\bz) \, & \frac{|z_\alpha|^2}{(h - E -\sum_{\alpha} d_{\alpha} |z_{\alpha}|^2)^2 + (\sum_\alpha c_{\alpha} |z_{\alpha}|^2 )^2}  \lesssim \min \left( \frac{\Delta}{c_\alpha} ,  \, \frac{\Delta^{7/8}}{c_\alpha |d_{\beta_0}|^{1/8}} , \frac{\Delta^3}{\eps} \right) \end{split} \end{equation}
\end{proposition}

\begin{proof}
In order to prove the first two bounds on the r.h.s. of (\ref{eq:propI}), we estimate
\begin{equation}\label{eq:I1} \text{A} \lesssim  \int
 d\mu (\bz) \, \frac{|z_\alpha|^2}{(h - E -\sum_{\alpha} d_{\alpha} |z_{\alpha}|^2)^2 + (c_{\alpha} |z_{\alpha}|^2 )^2} \end{equation}
We define the function \[ F (t) = \int_{-\infty}^t \frac{ds}{s^2 + (c_{\alpha}|z_{\alpha}|^2)^2} \, .  \]
Next, we make use of the indices $\beta_j$, $j=0,1,2,3$ defined in (\ref{eq:beta0}) and (\ref{eq:indices}). Defining the signs $\sigma_j$, $j=0,1,2,3$, by $\sigma_j = 1$, if $\lambda_{\beta_j} \geq E$, and $\sigma_j = -1$ if $\lambda_{\beta_j} < E$, we observe that
\[\begin{split}
\left(\sum_{j=0}^3 \sigma_j \, z_{\beta_j} \frac{d}{dz_{\beta_j}} \right) & F (h-E-\sum_{\alpha}d_{\alpha}|z_{\alpha}|^2) \\ = \;& -\frac{\sum_{j=0}^3 |d_{\beta_j}| |z_{\beta_j}|^2}{(h - E -\sum_{\alpha}d_{\alpha}|z_{\alpha}|^2)^2 + (c_{\alpha}|z_{\alpha}|^2)^2} \\ &-2 \sum_{j=0}^3 \delta_{\alpha,\beta_j} (c_\alpha |z_\alpha|^2)^2
\int_{-\infty}^{h-E-\sum_\alpha d_\alpha |z_\alpha|^2} ds
\, \frac{1}{\left(s^2 + (c_\alpha |z_\alpha|^2)^2 \right)^2} \, ,
\end{split} \]
where we used the fact that, by definition, $\sigma_j d_{\beta_j} = |d_{\beta_j}|$. Thus
\[ \begin{split}
&\frac{1}{(h - E -\sum_{\alpha}d_{\alpha}|z_{\alpha}|^2)^2 + (c_{\alpha}|z_{\alpha}|^2)^2} \\ &\hspace{2cm}=
\; - \frac{1}{\sum_{j=0}^3 |d_{\beta_j}| \,  |z_{\beta_j}|^2} \left( \sum_{j=0}^3 \sigma_j \, z_{\beta_j} \frac{d}{dz_{\beta_j}} \right) F (h-E-\sum_{\alpha}d_{\alpha}|z_{\alpha}|^2) \\ &\hspace{2.5cm} - 2 \sum_{j=0}^3 \sigma_j \delta_{\alpha,\beta_j} \frac{ (c_\alpha |z_\alpha|^2)^2}{\sum_{j=0}^3 |d_{\beta_j}| |z_{\beta_j}|^2} \int_{-\infty}^{h-E-\sum_\alpha d_\alpha |z_\alpha|^2} ds \frac{1}{\left(s^2 + (c_\alpha |z_\alpha|^2)^2 \right)^2}
\end{split} \]
When we insert this identity into (\ref{eq:I1}), we obtain
\[ \begin{split}
\text{A} \leq  \; & -
 \int  d\mu (\bz)  \,  \frac{|z_\alpha|^2}{\sum_{j=0}^3 |d_{\beta_j}| |z_{\beta_j}|^2} \left( \sum_{j=0}^3 \sigma_j z_{\beta_j} \frac{d}{dz_{\beta_j}} \right) F (h-E-\sum_{\alpha}d_{\alpha}|z_{\alpha}|^2)
\\ &- 2 \left(\sum_{j=0}^3 \sigma_j \delta_{\beta_j,\alpha}\right)
\int  d\mu(\bz) \,  \, |z_{\alpha}|^2
 \frac{(c_{\alpha} |z_{\alpha}|^2)^2}{\sum_{j=0}^3 |d_{\beta_j}| |z_{\beta_j}|^2} \int_{-
\infty}^{h-E-\sum_\alpha d_\alpha |z_\alpha|^2} ds \frac{1}{\left(s^2 + (c_{\alpha} |z_{\alpha}|^2)^2 \right)^2} \\ =: \; & \text{A}_1 + \text{A}_2
\end{split} \]
Since \[
\int_{-\infty}^{\infty}  \frac{1}{\left(s^2 + (c_{\alpha} |z_{\alpha}|^2)^2 \right)^2} \lesssim \frac{1}{(c_{\alpha} |z_{\alpha}|^2)^3} \]
 we find that
\begin{equation}\label{eq:bdI2} | \text{A}_2 | \lesssim  \, \frac{1}{c_{\alpha}} \,
\int  d\mu(\bz) \,  \frac{1}{\sum_{j=0}^3 |d_{\beta_j}| |z_{\beta_j}|^2} \, . \end{equation}
As for the term $\text{A}_1$, we integrate by parts. Introducing the function \begin{equation}\label{eq:phi} \phi (\bz) = \sum_{j=1}^{N-1} g (\text{Re } (U\bz)_j) + g (\text{Im } (U\bz)_j) \end{equation} we have $d\mu (\bz) = e^{-\phi (\bz)} d\bz d\bar{\bz}$ (recall the definition (\ref{eq:dmu}), and the fact that $h = e^{-g}$). Therefore, we obtain that
\begin{equation}\label{eq:ibp} \begin{split}
\text{A}_1 = \; & \,
 \int  d\mu (\bz) \, \left[  \sum_{j=0}^3 \sigma_j  \frac{d}{dz_{\beta_j}} \,  \frac{z_{\beta_j} \, |z_\alpha|^2}{\sum_{j=0}^3 d_{\beta_j} |z_{\beta_j}|^2} \right] \, F
 (h-E-\sum_{\alpha}d_{\alpha}|z_{\alpha}|^2) \\
&+
 \int  d\bz \, d {\bar{\bz}} \, \left[ \left( \sum_{j=0}^3 z_{\beta_j} \frac{d}{dz_{\beta_j}} \right) e^{-\phi (\bz)} \right]\, \frac{|z_\alpha|^2}{\sum_{j=0}^3 d_{\beta_j} |z_{\beta_j}|^2} \,  F (h-E-\sum_{\alpha}d_{\alpha}|z_{\alpha}|^2) \, .
\end{split} \end{equation}
Simple computation shows that, for any $\alpha$,
\begin{equation}\label{eq:der-bd} \left| \sum_{j=0}^3 \sigma_j \frac{d}{dz_{\beta_j}}\frac{z_{\beta_j} |z_\alpha|^2}{\sum_{j=0}^3 |d_{\beta_j}| |z_{\beta_j}|^2} \right| \lesssim \frac{|z_\alpha|^2}{\sum_{j=0}^3 |d_{\beta_j}| |z_{\beta_j}|^2}\end{equation}
Since
\[ 0 \leq F (h-E-\sum_{\alpha}d_{\alpha}|z_{\alpha}|^2) \lesssim \frac{1}{c_\alpha |z_\alpha|^2} \]
we conclude that
\[ |\text{A}_1| \lesssim \, \frac{1}{c_\alpha}
 \int  d\mu(\bz) \, \frac{1}{\sum_{j=0}^3 |d_{\beta_j}| |z_{\beta_j}|^2} \,  \left( 1 + \sum_{j=0}^3 |z_{\beta_j}| \left| \frac{d\phi (\bz)}{d z_{\beta_j}} \right| \right) \]
Combining with (\ref{eq:bdI2}), we obtain
\begin{equation} \label{eq:intbypa}
\text{A} \lesssim \, \frac{1}{c_\alpha}
 \int  d\mu (\bz) \, \frac{1}{\sum_{j=0}^3 |d_{\beta_j}| |z_{\beta_j}|^2} \,  \left( 1 + \sum_{j=0}^3 |z_{\beta_j}| \left| \frac{d\phi (\bz)}{d z_{\beta_j}} \right| \right)
\end{equation}
If we neglect the term with $j=0$ in the denominator, we find
\[ \begin{split} \text{A} \lesssim \; &\frac{1}{c_\alpha} \frac{1}{\min (|d_{\beta_1}|, |d_{\beta_2}| , |d_{\beta_3}|)} \int d\mu (\bz) \frac{1}{\sum_{j=1}^3 |z_{\beta_j}|^2} \left( 1 + \sum_{j=0}^3 |z_{\beta_j}| \left| \frac{d\phi (\bz)}{d z_{\beta_j}} \right| \right) \\
\lesssim \; &\frac{\Delta}{c_\alpha} \int d\mu (\bz) \frac{1}{|z_{\beta_1}|^2 + |z_{\beta_2}|^2} \\ &+\frac{\Delta}{c_\alpha} \sum_{j=0}^3 \left( \int d\mu(\bz) \, \left| \frac{d\phi}{dz_{\beta_j}} \right|^4 \right)^{1/4} \, \left( \int d\mu (\bz) \, |z_{\beta_j}|^4 \right)^{1/4}  \left( \int d\mu (\bz) \, \frac{1}{(|z_{\beta_1}|^2 + |z_{\beta_2}|^2 + |z_{\beta_3}|^2)^2} \right)^{1/2}
\\ \lesssim \; & \frac{\Delta}{c_\alpha}
\end{split} \]
where we used Lemma \ref{lm:pow}, Lemma \ref{lm:dphi}, Lemma \ref{lm:CS} and the fact that, from (\ref{eq:delta}), $|d_{\beta_j}| \geq \Delta$ for $j=1,2,3$. Similarly, starting from (\ref{eq:intbypa}), we also deduce that
\[ \begin{split}
 \text{A} \lesssim \; &\frac{\Delta^{7/8}}{c_\alpha |d_{\beta_0}|^{1/8}}  \int d\mu (\bz) \, \frac{1}{|z_{\beta_0}|^{1/4}} \frac{1}{\left(\sum_{j=1}^4 |z_{\beta_j}|^2 \right)^{7/8}} \left( 1 + \sum_{j=0}^4 |z_{\beta_j}| \left| \frac{d\phi (\bz)}{d z_{\beta_j}} \right| \right) \\
\lesssim \; &\frac{\Delta^{7/8}}{c_\alpha |d_{\beta_0}|^{1/8}} \left( \int d\mu (\bz) \frac{1}{|z_{\beta_0}|} \right)^{1/4} \left( \int d\mu (\bz) \, \frac{1}{\left( |z_{\beta_1}|^2 + |z_{\beta_2}|^2 \right)^{7/6}} \right)^{3/4} \\ &+ \frac{\Delta^{7/8}}{c_\alpha |d_{\beta_0}|^{1/8}} \sum_{j=1}^4 \left( \int d\mu(\bz) \, \left| \frac{d\phi}{dz_{\beta_j}} \right|^4 \right)^{1/4} \, \left( \int d\mu (\bz) \frac{1}{|z_{\beta_0}|} \right)^{1/4} \, \left( \int d\mu (\bz) \, |z_{\beta_j}|^{16} \right)^{1/16} \\ &\hspace{5cm} \times  \left( \int d\mu (\bz) \, \frac{1}{(|z_1|^2 + |z_2|^2 + |z_3|^2 )^2} \right)^{7/16}
\\ \lesssim \; & \frac{\Delta^{7/8}}{c_\alpha |d_{\beta_0}|^{1/8}} \, .
\end{split} \]
In order to obtain the third bound in (\ref{eq:propI}), we make use of the indices $\beta_j$, $j=1, \dots, 8$, defined in (\ref{eq:indices}). As above, we define the signs $\sigma_j$, $j=1,\dots,8$, by $\sigma_j = 1$, if $\lambda_{\beta_j} \geq E$, and $\sigma_j = -1$ if $\lambda_{\beta_j} < E$. We estimate
\begin{equation}\label{eq:Asec} \text{A} \leq \int
 d\mu (\bz) \, \frac{|z_\alpha|^2}{(h - E -\sum_{\alpha} d_{\alpha} |z_{\alpha}|^2)^2 + \left( \sum_{j=5}^8 c_{\beta_j} |z_{\beta_j}|^2 \right)^2} \end{equation}
Defining the function \[ G (t) = \int_{-\infty}^t \frac{ds}{s^2 + \left( \sum_{j=5}^8 c_{\beta_j} |z_{\beta_j}|^2 \right)^2} \, ,  \]
we observe that
\[\begin{split}
\left(\sum_{j=1}^4 \sigma_j \, z_{\beta_j} \frac{d}{dz_{\beta_j}} \right) & G (h-E-\sum_{\alpha}d_{\alpha}|z_{\alpha}|^2) \\ &\hspace{2cm} = -\frac{\sum_{j=1}^4 |d_{\beta_j}| |z_{\beta_j}|^2}{(h - E -\sum_{\alpha}d_{\alpha}|z_{\alpha}|^2)^2 + \left( \sum_{j=5}^8 c_{\beta_j} |z_{\beta_j}|^2 \right)^2} \, .
\end{split} \]
Thus, integrating by parts,
\[ \begin{split}
\text{A} \lesssim &\; \int  d\mu (\bz)  \,  \frac{|z_\alpha|^2}{\sum_{j=1}^4 |d_{\beta_j}| |z_{\beta_j}|^2} \left( \sum_{j=1}^4 \sigma_j \, z_{\beta_j} \frac{d}{dz_{\beta_j}} \right) G (h-E-\sum_{\alpha}d_{\alpha}|z_{\alpha}|^2)  \\
\lesssim &\; \int  d\mu (\bz) \,  \left[ \sum_{j=1}^4 \sigma_j \frac{d}{dz_{\beta_j}} \frac{z_{\beta_j} |z_\alpha|^2}{\sum_{j=1}^4 |d_{\beta_j}| |z_{\beta_j}|^2} \right] \, G (h-E-\sum_{\alpha}d_{\alpha}|z_{\alpha}|^2) \\
&+   \int  d\bz \, d {\bar{\bz}} \, \left[ \left( \sum_{j=1}^4 z_{\beta_j} \frac{d}{dz_{\beta_j}} \right) e^{-\phi (\bz)} \right]\, \frac{|z_\alpha|^2}{\sum_{j=1}^4 |d_{\beta_j}| |z_{\beta_j}|^2} \,  G (h-E-\sum_{\alpha}d_{\alpha}|z_{\alpha}|^2)
\end{split} \]
Using a bound analogous to (\ref{eq:der-bd}) and
\[ G(h-E-\sum_{\alpha}d_{\alpha}|z_{\alpha}|^2) \leq \frac{1}{\sum_{j=5}^8 c_{\beta_j} |z_{\beta_j}|^2} \]
we obtain, since $|d_{\beta_j}| \geq \Delta^{-1}$ for  $j=1,\dots, 4$ and $c_{\beta_j} \geq \eps \, \Delta^{-2}$ for $j=5,\dots ,8$ (by (\ref{eq:delta}) and (\ref{eq:cdelta})),
\[ \begin{split} \text{A} \lesssim \; & \frac{\Delta^3}{\eps} \int d\mu(\bz) \frac{|z_{\alpha}|^2}{\left(\sum_{j=1}^4 |z_{\beta_j}|^2 \right) \left( \sum_{j=5}^8 |z_{\beta_j}|^2 \right)} \, \left(1 + \sum_{j=1}^4 |z_{\beta_j}| \left| \frac{d \phi (\bz)}{dz_{\beta_j}} \right| \right) \\
 \lesssim \; & \frac{\Delta^3}{\eps} \left( \int d\mu (\bz) \, |z_\alpha|^6 \right)^{1/3} \left( \int d\mu(\bz) \frac{1}{\left(\sum_{j=1}^4 |z_{\beta_j}|^2 \right)^3} \right)^{1/3} \ \left( \int d\mu (\bz) \, \frac{1}{ \left( \sum_{j=5}^8 |z_{\beta_j}|^2 \right)^3} \right)^{1/3} \\
 &+ \frac{\Delta^3}{\eps} \left( \int d\mu (\bz) \, |z_\alpha|^6 \right)^{1/3} \, \left( \int d\mu(\bz) \, \left| \frac{d\phi}{dz_{\beta_j}} \right|^6 \right)^{1/6} \left( \int d\mu(\bz) \frac{1}{\left(\sum_{j=1}^4 |z_{\beta_j}|^2 \right)^3} \right)^{1/6} \\ &\hspace{5cm} \times  \left( \int d\mu (\bz) \, \frac{1}{ \left( \sum_{j=5}^8 |z_{\beta_j}|^2 \right)^3} \right)^{1/3}  \end{split} \]
Lemma \ref{lm:pow}, Lemma \ref{lm:dphi} and Lemma \ref{lm:CS} imply that $A \lesssim \Delta^3 / \eps$. This concludes the proof of (\ref{eq:propI}).
\end{proof}

The following proposition is used in the analysis of the term (\ref{eq:II}).

\begin{proposition}\label{prop:II}
Let $H$ be an ensemble of Hermitian Wigner matrices as in Definition \ref{def}, so that $\E \, e^{\nu |x_{ij}|^2} < \infty$ for some $\nu >0$. Let real and imaginary part of the off-diagonal entries have a common probability density function $h$ such that (\ref{eq:ass2}) holds true. Let $B$ be the $(N-1) \times (N-1)$ minor of $H$ obtained by removing the first row and the first column of $H$. Suppose the randomness in $B$ is such that  the event $\Omega$, defined in (\ref{eq:Omega}), is satisfied. Let the measure $d\mu (\bz)$ be defined as in (\ref{eq:dmu}). Then we have, for every $\alpha = 1, \dots , N-1$,
\begin{equation}\label{eq:propII} \begin{split}& \left| \int  d\mu (\bz) \,   |z_{\alpha}|^2  \,   \frac{\left(\frac{\eps}{N} + \sum_\gamma c_\gamma |z_{\gamma}|^2\right)  \, \left(h - E -\sum_{\alpha}d_{\alpha}|z_\alpha|^2\right)}{\left[(h - E -\sum_{\alpha}d_{\alpha}|z_{\alpha}|^2)^2 + (\frac{\eps}{N} + \sum_{\alpha}c_{\alpha}|z_{\alpha}|^2)^2\right]^2} \right| \\ &\hspace{8cm} \lesssim \min \left(\frac{\Delta}{|d_\alpha|} , \, \frac{\Delta^{7/8}}{|d_\alpha | |d_{\beta_0}|^{1/8}} ,  \frac{\Delta}{c_{\alpha}}, \Delta^2 \right) \end{split} \end{equation}
Moreover, we have
\begin{equation}\label{eq:propIII}
 \left| \int  d\mu (\bz)  \,   \frac{\left(\frac{\eps}{N} + \sum_\gamma c_\gamma |z_{\gamma}|^2\right) \, \left(h - E -\sum_{\alpha}d_{\alpha}|z_\alpha|^2\right)}{\left[(h - E -\sum_{\alpha}d_{\alpha}|z_{\alpha}|^2)^2 + (\frac{\eps}{N} + \sum_{\alpha}c_{\alpha}|z_{\alpha}|^2)^2\right]^2} \right| \lesssim \Delta^2
 \end{equation}
\end{proposition}

\begin{proof}
We first show (\ref{eq:propII}). Let \[ \text{B} := \int  d\mu (\bz) \,   |z_{\alpha}|^2  \, \left(\frac{\eps}{N} + \sum_\gamma c_\gamma |z_{\gamma}|^2\right)   \frac{h - E -\sum_{\alpha}d_{\alpha}|z_\alpha|^2}{\left[(h - E -\sum_{\alpha}d_{\alpha}|z_{\alpha}|^2)^2 + (\frac{\eps}{N} + \sum_{\alpha}c_{\alpha}|z_{\alpha}|^2)^2\right]^2} \, . \]
We start by observing that
\begin{equation} \begin{split} z_\alpha \, \frac{d}{dz_\alpha} &\frac{1}{
(h - E -\sum_{\alpha}d_{\alpha}|z_{\alpha}|^2)^2 + (\frac{\eps}{N} + \sum_{\alpha}c_{\alpha}|z_{\alpha}|^2)^2} \\ &\hspace{2cm} = \;  - 2 d_\alpha |z_\alpha|^2 \, \frac{(h - E -\sum_{\alpha}d_{\alpha}|z_\alpha|^2)}{\left[(h - E -\sum_{\alpha}d_{\alpha}|z_{\alpha}|^2)^2 + (\frac{\eps}{N} + \sum_{\alpha}c_{\alpha}|z_{\alpha}|^2)^2\right]^2} \\ &\hspace{2.5cm} -2c_\alpha |z_\alpha|^2 \frac{(\eps/N + \sum_\alpha c_\alpha |z_\alpha|^2)}{\left[(h - E -\sum_{\alpha}d_{\alpha}|z_{\alpha}|^2)^2 + (\frac{\eps}{N} + \sum_{\alpha}c_{\alpha}|z_{\alpha}|^2)^2\right]^2} \,. \end{split} \end{equation}
Therefore we obtain that
\begin{equation}\label{eq:II1} \begin{split} \text{B} =  \; &-\frac{1}{2d_\alpha} \int d\mu (\bz)  \,  \left(\frac{\eps}{N} +  \sum_\gamma c_\gamma |z_{\gamma}|^2 \right) \, z_\alpha \frac{d}{dz_\alpha} \frac{1}{(h - E -\sum_{\alpha}d_{\alpha}|z_{\alpha}|^2)^2 + (\frac{\eps}{N} + \sum_{\alpha}c_{\alpha}|z_{\alpha}|^2)^2}  \\
&- \frac{1}{d_\alpha} \,\int d\mu (\bz) \, (c_\alpha |z_\alpha|^2)  \frac{\left(\frac{\eps}{N} + \sum_\gamma c_\gamma |z_{\gamma}|^2\right)^2 }{\left[(h - E -\sum_{\alpha}d_{\alpha}|z_{\alpha}|^2)^2 + (\frac{\eps}{N} + \sum_{\alpha}c_{\alpha}|z_{\alpha}|^2)^2\right]^2}
\\ =: \; & \text{B}_1 + \text{B}_2
\end{split} \end{equation}
The abolute value of the term $\text{B}_2$ can be bounded by
\[ \begin{split}
|\text{B}_2| \leq \frac{c_\alpha}{|d_\alpha|} \,\int d\mu (\bz) \,  \frac{|z_\alpha|^2}{(h - E -\sum_{\alpha}d_{\alpha}|z_{\alpha}|^2)^2 + (c_{\alpha}|z_{\alpha}|^2)^2} \,.
\end{split}\]
Eq. (\ref{eq:propI}) implies that
\begin{equation}\label{eq:B2-final} |\text{B}_2| \leq \min \left(\frac{\Delta}{|d_\alpha|} , \frac{\Delta^{7/8}}{|d_\alpha|\, |d_{\beta_0}|^{1/8}} \right)  \end{equation}
To control the contribution $\text{B}_1$ in (\ref{eq:II1}) we integrate by parts:\begin{equation}\label{eq:A} \begin{split} \text{B}_1 = \frac{1}{2d_\alpha} \int  d\bz d\bar{\bz}& \, \frac{d}{dz_\alpha} \left[ z_\alpha e^{-\phi (\bz)} \left( \frac{\eps}{ N} + \sum_\gamma c_\gamma |z_\gamma|^2 \right) \right] \\ &\hspace{2cm}\times  \frac{1}{(h - E -\sum_{\alpha}d_{\alpha}|z_{\alpha}|^2)^2 + (\frac{\eps}{N} + \sum_{\alpha}c_{\alpha}|z_{\alpha}|^2)^2} \, . \end{split} \end{equation}
Next, we make use of the indices $\beta_j$, $j=0,\dots,4$ defined in (\ref{eq:beta0}) and (\ref{eq:indices}). As in the proof of Proposition \ref{prop:I}, we introduce the signs $\sigma_j$, $j=0,\dots,4$, by $\sigma_j = 1$, if $\lambda_{\beta_j} \geq E$, and $\sigma_j = -1$ if $\lambda_{\beta_j} < E$. We define next the function
\begin{equation}\label{eq:L}  L(t) = \int_{-\infty}^t ds \frac{1}{s^2+ \left( \eps/ N + \sum_\gamma c_\gamma |z_\gamma|^2 \right)^2} \end{equation}
and we observe that
\[ \begin{split}
&\frac{1}{(h - E -\sum_{\alpha}d_{\alpha}|z_{\alpha}|^2)^2 + (\frac{\eps}{N} + \sum_\gamma c_\gamma |z_\gamma|^2 )^2} \\ &\hspace{0cm}=
\; - \frac{1}{\sum_{j=0}^4 |d_{\beta_j}| |z_{\beta_j}|^2} \left( \sum_{j=0}^4 \sigma_j \, z_{\beta_j} \frac{d}{dz_{\beta_j}} \right) L (h-E-\sum_{\alpha}d_{\alpha}|z_{\alpha}|^2) \\ &\hspace{0.5cm} - \frac{2 \left(\eps/N + \sum_\alpha c_\alpha |z_\alpha|^2 \right) \, (\sum_{j=0}^4 \sigma_j \, c_{\beta_j} |z_{\beta_j}|^2)}{\sum_{j=0}^4 |d_{\beta_j}| |z_{\beta_j}|^2} \int_{-\infty}^{h-E-\sum_\alpha d_\alpha |z_\alpha|^2} ds \frac{1}{\left(s^2 + (\frac{\eps}{N} + \sum_\alpha c_\alpha |z_\alpha|^2)^2 \right)^2}
\end{split} \]
Inserting this expression into (\ref{eq:A}), we find
\begin{equation}\label{eq:A1A2}  \begin{split} \text{B}_1 = \; & -  \frac{1}{2d_\alpha} \int  d\bz d\bar{\bz} \, \frac{d}{dz_\alpha} \left[ z_\alpha e^{-\phi (\bz)} \left( \frac{\eps}{N} + \sum_\gamma c_\gamma |z_\gamma|^2 \right) \right] \\ &\hspace{.5cm}\times \frac{1}{\sum_{j=0}^4 |d_{\beta_j}| |z_{\beta_j}|^2} \left( \sum_{j=0}^4 \sigma_j \, z_{\beta_{j}} \frac{d}{dz_{\beta_{j}}} \right) L (h-E-\sum_{\alpha}d_{\alpha}|z_{\alpha}|^2) \\ &- \frac{1}{d_\alpha} \int  d\bz d\bar{\bz}\frac{d}{dz_\alpha} \left[ z_\alpha e^{-\phi (\bz)} \left( \frac{\eps}{N} + \sum_\gamma c_\gamma |z_\gamma|^2 \right) \right] \\ &\hspace{.5cm} \times  \frac{\left(\eps/N + \sum_\alpha c_\alpha |z_\alpha|^2 \right) \, (\sum_{j=0}^4 \sigma_j \, c_{\beta_j} |z_{\beta_j}|^2)}{\sum_{j=0}^4 |d_{\beta_j}| |z_{\beta_j}|^2} \int_{-\infty}^{h-E-\sum_\alpha d_\alpha |z_\alpha|^2} \frac{ds}{\left(s^2 + (\frac{\eps}{N} + \sum_\alpha c_\alpha |z_\alpha|^2)^2 \right)^2} \\
 =: \; &\text{B}_3 + \text{B}_4
 \end{split} \end{equation}
 Using the bound
\[ \int_{-\infty}^{\infty} \frac{ds}{\left(s^2 + (\frac{\eps}{N} + \sum_\alpha c_\alpha |z_\alpha|^2)^2\right)^2} \lesssim \frac{1}{(\frac{\eps}{N} + \sum_\alpha c_\alpha |z_\alpha|^2)^3} \]
we conclude that
\begin{equation}\label{eq:A2-final}
|\text{B}_4| \lesssim \frac{1}{|d_\alpha|} \int d\mu (\bz) \, \frac{1}{\sum_{j=0}^4 |d_{\beta_j}| |z_{\beta_j}|^2} \left(1+ |z_\alpha| \left|\frac{d\phi (\bz)}{dz_\alpha}\right| \right) \end{equation}
As for the term $\text{B}_3$ in (\ref{eq:A1A2}), we integrate again by parts. Taking absolute value after integration by parts, and using (\ref{eq:der-bd}), we find
\[ \begin{split}
|\text{B}_3|  \lesssim & \; \frac{1}{|d_\alpha|} \int  d\bz d\bar{\bz} \, \left| \frac{d}{dz_\alpha} \left[ z_\alpha e^{-\phi (\bz)} \left( \frac{\eps}{N} + \sum_\gamma c_\gamma |z_\gamma|^2 \right) \right] \right| \, \frac{L (h-E-\sum_{\alpha}d_{\alpha}|z_{\alpha}|^2)}{\sum_{j=0}^4 |d_{\beta_j}| |z_{\beta_j}|^2} \\
&+ \frac{1}{|d_\alpha|} \sum_{j=0}^4 \int  d\bz d\bar{\bz} \, |z_{\beta_j}| \, \left|\frac{d^2}{dz_{\beta_j} dz_\alpha} \left[ z_\alpha e^{-\phi (\bz)} \left( \frac{\eps}{N} + \sum_\gamma c_\gamma |z_\gamma|^2 \right) \right] \right| \, \frac{L (h-E-\sum_{\alpha}d_{\alpha}|z_{\alpha}|^2)}{\sum_{j=0}^4 |d_{\beta_j}| |z_{\beta_j}|^2}
\end{split} \]
Computing the derivatives and using the bound
\[ 0 \leq L( h-E-\sum_{\alpha}d_{\alpha}|z_{\alpha}|^2) \leq \frac{1}{\frac{\eps}{N} + \sum_\gamma c_\gamma |z_\gamma|^2}, \]
we arrive at
\[ \begin{split}
|\text{B}_3| \lesssim \; & \frac{1}{|d_\alpha|} \int d\mu (\bz) \, \frac{1}{\sum_{j=0}^4 |d_{\beta_j}| |z_{\beta_j}|^2}  \left( 1+ |z_\alpha| \left| \frac{d\phi (\bz)}{dz_\alpha}\right| + \sum_{j=0}^4 |z_{\beta_j}|  \left| \frac{d\phi (\bz)}{dz_{\beta_j}}\right|  \right. \\ &\hspace{4cm} \left. + \sum_{j=0}^4 |z_\alpha| |z_{\beta_j}| \left| \frac{d^2 \phi (\bz)}{dz_{\beta_j} dz_\alpha} \right| + \sum_{j=0}^4 |z_\alpha| |z_{\beta_j}| \left| \frac{d\phi (\bz)}{dz_\alpha} \right|\, \left| \frac{d\phi (\bz)}{dz_{\beta_j}} \right| \right)
 \end{split} \]
Combining this estimate with (\ref{eq:A2-final}), we find, from (\ref{eq:A1A2}),
\begin{equation} \label{eq:B1bd} \begin{split}
|\text{B}_1| \lesssim\;  & \frac{1}{|d_\alpha|} \int d\mu (\bz) \, \frac{1}{\sum_{j=0}^4 |d_{\beta_j}| |z_{\beta_j}|^2} \\ &\hspace{1cm} \times \left( 1+ |z_\alpha|^2 \left| \frac{d\phi (\bz)}{dz_\alpha}\right|^2 + \sum_{j=0}^3 |z_{\beta_j}|^2  \left| \frac{d\phi (\bz)}{dz_{\beta_j}}\right|^2  + \sum_{j=0}^3 |z_\alpha| |z_{\beta_j}| \left| \frac{d^2 \phi (\bz)}{dz_{\beta_j} dz_\alpha} \right| \right)
 \end{split} \end{equation}
If we neglect the term with $j=0$ in the denominator, we obtain (since $|d_{\beta_j}| \geq 1/(2\Delta)$ for $j=1,\dots,4$ by (\ref{eq:delta})),
\begin{equation} \label{eq:B1-1}\begin{split} |\text{B}_1| \lesssim \; & \frac{\Delta}{|d_\alpha|}  \int d\mu (\bz) \, \frac{1}{\sum_{j=1}^4  |z_{\beta_j}|^2} \left( 1+ |z_\alpha|^2 \left| \frac{d\phi (\bz)}{dz_\alpha}\right|^2 + \sum_{j=0}^4 |z_{\beta_j}|^2  \left| \frac{d\phi (\bz)}{dz_{\beta_j}}\right|^2  + \sum_{j=0}^4 |z_\alpha| |z_{\beta_j}| \left| \frac{d^2 \phi (\bz)}{dz_{\beta_j} dz_\alpha} \right| \right) \\
\lesssim \; & \frac{\Delta}{|d_\alpha|}  \, \left[ \left( \int d\mu (\bz) \frac{1}{|z_{\beta_1}|^2 + |z_{\beta_2}|^2} \right)  \right. \\ & \hspace{1cm} + \sum_{\gamma=\alpha,\beta_1,..,\beta_4} \left(
\int d\mu (\bz) \, \left| \frac{d\phi}{dz_\gamma} \right|^6 \right)^{1/3} \left( \int d\mu (\bz) \frac{1}{\left( \sum_{j=1}^4 |z_{\beta_j}|^2 \right)^3} \right)^{1/3} \left( \int d\mu (\bz) \, |z_\gamma|^6 \right)^{1/3}  \\ &\hspace{1cm}  + \sum_{j=0}^4 \left(\int d\mu(\bz) \left| \frac{d^2\phi}{dz_\alpha dz_{\beta_j}} \right|^2 \right)^{1/2} \,
\left( \int d\mu (\bz) \frac{1}{\left( \sum_{j=1}^4 |z_{\beta_j}|^2 \right)^3} \right)^{1/3}
\\ &\left. \hspace{5cm} \times \left( \int d\mu (\bz) \, |z_\alpha|^{12} \right)^{1/12} \left( \int d\mu (\bz) \, |z_{\beta_j}|^{12} \right)^{1/12}
\right]
\end{split} \end{equation}
Applying Lemma \ref{lm:pow}, Lemma \ref{lm:dphi} and Lemma \ref{lm:CS}, we conclude that $|\text{B}_1| \lesssim \Delta / |d_\alpha|$.
On the other hand, from (\ref{eq:B1bd}), we also conclude that
\[ \begin{split} |\text{B}_1| \lesssim & \; \frac{\Delta^{7/8}}{|d_\alpha| |d_{\beta_0}|^{1/8}} \int d\mu (\bz) \, \frac{1}{|z_{\beta_0}|^{1/4}} \, \frac{1}{\left( \sum_{j=1}^4 |z_{\beta_j}|^2 \right)^{7/8}} \\ &\hspace{1cm} \times \left( 1+ |z_\alpha|^2 \left| \frac{d\phi (\bz)}{dz_\alpha}\right|^2 + \sum_{j=0}^3 |z_{\beta_j}|^2  \left| \frac{d\phi (\bz)}{dz_{\beta_j}}\right|^2  + \sum_{j=0}^3 |z_\alpha| |z_{\beta_j}| \left| \frac{d^2 \phi (\bz)}{dz_{\beta_j} dz_\alpha} \right| \right)
\\ \lesssim &\; \frac{\Delta^{7/8}}{|d_\alpha| |d_{\beta_0}|^{1/8}} \left[ \left( \int d\mu (\bz) \, \frac{1}{|z_{\beta_0}|} \right)^{1/4} \, \left( \int d\mu (\bz) \frac{1}{\left(|z_{\beta_1}|^2 + |z_{\beta_2}|^2 \right)^{7/6}} \right)^{3/4}\right. \\ & \hspace{.5cm} +\sum_{\gamma=\alpha,\beta_1,..,\beta_4} \left(
\int d\mu (\bz) \, \left| \frac{d\phi}{dz_\gamma} \right|^6 \right)^{1/3} \left( \int d\mu (\bz) \frac{1}{\left( \sum_{j=1}^4 |z_{\beta_j}|^2 \right)^3} \right)^{7/24} \\ & \hspace{5cm} \times \left( \int d\mu (\bz) \, \frac{1}{|z_{\beta_0}|^{3/2}} \right)^{1/6} \left( \int  d\mu (\bz) \, |z_\gamma|^{48/5} \right)^{5/24} \\
\\ & \hspace{.5cm}  +\sum_{j=0}^4 \left(
\int d\mu (\bz) \, \left| \frac{d^2\phi}{dz_\alpha dz_{\beta_j}} \right|^2 \right)^{1/2} \left( \int d\mu (\bz) \frac{1}{\left( \sum_{j=1}^4 |z_{\beta_j}|^2 \right)^3} \right)^{7/24} \\ & \hspace{1cm}\times  \left( \int d\mu (\bz) \, \frac{1}{|z_{\beta_0}|^{3/2}} \right)^{1/6} \left( \int  d\mu (\bz) \, |z_\alpha|^{48} \right)^{1/48} \left.   \left( \int  d\mu (\bz) \, |z_{\beta_j}|^{48} \right)^{1/48} \right] \\
\lesssim & \; \frac{\Delta^{7/8}}{|d_\alpha| |d_{\beta_0}|^{1/8}}
\end{split} \]
Together with (\ref{eq:B2-final}), we find
\[ |\text{B}| \leq \min \left(\frac{\Delta}{|d_\alpha|} , \frac{\Delta^{7/8}}{|d_\alpha|\, |d_{\beta_0}|^{1/8}} \right)  \]

In order to show the third and the fourth bound on the r.h.s. of (\ref{eq:propII}), we make use of the indices $\beta_j$, $j=1,\dots,8$ introduced in (\ref{eq:indices}). We observe that
\begin{equation}\begin{split} \Big( \sum_{j=1}^4 \sigma_j \, z_{\beta_j} &\frac{d}{dz_{\beta_j}} \Big) \,  \frac{1}{
(h - E -\sum_{\alpha}d_{\alpha}|z_{\alpha}|^2)^2 + (\frac{\eps}{N} + \sum_{\alpha}c_{\alpha}|z_{\alpha}|^2)^2} \\ &\hspace{.5cm} = \;  - 2 \left(\sum_{j=1}^4 |d_{\beta_j}| |z_{\beta_j}|^2 \right) \, \frac{(h - E -\sum_{\alpha}d_{\alpha}|z_\alpha|^2)}{\left[(h - E -\sum_{\alpha}d_{\alpha}|z_{\alpha}|^2)^2 + (\frac{\eps}{N} + \sum_{\alpha}c_{\alpha}|z_{\alpha}|^2)^2\right]^2} \\ &\hspace{1cm} -2\left( \sum_{j=1}^4 \sigma_j \,  c_{\beta_j} |z_{\beta_j}|^2 \right) \,  \frac{(\eps/N + \sum_\alpha c_\alpha |z_\alpha|^2)}{\left[(h - E -\sum_{\alpha}d_{\alpha}|z_{\alpha}|^2)^2 + (\frac{\eps}{N} + \sum_{\alpha}c_{\alpha}|z_{\alpha}|^2)^2\right]^2} \,. \end{split} \end{equation}
Therefore we obtain that
\begin{equation}\label{eq:B5B6} \begin{split} \text{B} =  \; &-\frac{1}{2} \int d\mu (\bz)  \,  \left(\frac{\eps}{N} +  \sum_\gamma c_\gamma |z_{\gamma}|^2 \right) \, \frac{|z_\alpha|^2}{\sum_{j=1}^4 |d_{\beta_j}| |z_{\beta_j}|^2} \\ &\hspace{2cm} \times  \left( \sum_{j=1}^4 \sigma_j \, z_{\beta_j} \frac{d}{dz_{\beta_j}} \right) \, \frac{1}{
(h - E -\sum_{\alpha}d_{\alpha}|z_{\alpha}|^2)^2 + (\frac{\eps}{N} + \sum_{\alpha}c_{\alpha}|z_{\alpha}|^2)^2}  \\
&-\,\int d\mu (\bz) \, \frac{|z_\alpha|^2 \left( \sum_{j=1}^4 \sigma_j \, c_{\beta_j} |z_{\beta_j}|^2 \right)}{\sum_{j=1}^4 |d_{\beta_j}| |z_{\beta_j}|^2}  \frac{\left(\frac{\eps}{N} + \sum_\gamma c_\gamma |z_{\gamma}|^2\right)^2 }{\left[(h - E -\sum_{\alpha}d_{\alpha}|z_{\alpha}|^2)^2 + (\frac{\eps}{N} + \sum_{\alpha}c_{\alpha}|z_{\alpha}|^2)^2\right]^2}
\\ =: \; & \text{B}_5 + \text{B}_6
\end{split} \end{equation}
The absolute value of $\text{B}_6$ can be bounded by
\begin{equation}\label{eq:B6-I}
|\text{B}_6| \lesssim \frac{1}{c_\alpha} \int d\mu (\bz) \, \frac{1}{\sum_{j=1}^4 |d_{\beta_j}| |z_{\beta_j}|^2} \lesssim \frac{\Delta}{c_\alpha}  \int d\mu (\bz) \, \frac{1}{|z_{\beta_1}|^2 + |z_{\beta_2}|^2} \lesssim  \frac{\Delta}{c_{\alpha}} \end{equation}
where we used Lemma \ref{lm:CS}. Alternatively, we can estimate
\[ |\text{B}_6| \leq \int d\mu(\bz) \frac{|z_\alpha|^2 \, (\sum_{j=1}^4 c_{\beta_j} |z_{\beta_j}|^2)}{\sum_{j=1}^4 |d_{\beta_j}| |z_{\beta_j}|^2} \, \frac{1}{(h - E -\sum_{\alpha}d_{\alpha}|z_{\alpha}|^2)^2 + (\sum_{j=1}^4 c_{\beta_j} |z_{\beta_j}|^2)^2}
\]
Since (recall $\sigma_j =1$ if $\lambda_{\beta_j} \geq E$, $\sigma_j = -1$ if $\lambda_{\beta_j} <E$)
\[ \begin{split} &\frac{1}{(h - E -\sum_{\alpha}d_{\alpha}|z_{\alpha}|^2)^2 + (\sum_{j=1}^4 c_{\beta_j} |z_{\beta_j}|^2)^2} \\ & \hspace{4cm} = \frac{1}{\sum_{j=5}^8 |d_{\beta_j}| |z_{\beta_j}|^2} \left( \sum_{j=5}^8 \sigma_j \, z_{\beta_j} \frac{d}{dz_{\beta_j}} \right) M (h-E -\sum_\alpha d_\alpha |z_\alpha|^2) \end{split} \]
with \[ M(t) = \int_{-\infty}^t ds \frac{1}{s^2 +  (\sum_{j=1}^4 c_{\beta_j} |z_{\beta_j}|^2)^2} \] we conclude integrating by parts and estimating all terms by their absolute value that
\[\begin{split} |\text{B}_6| \leq \; & \int d\mu (\bz) \, \frac{|z_\alpha|^2 (\sum_{j=1}^4 c_{\beta_j} |z_{\beta_j}|^2)}{ (\sum_{j=1}^4 |d_{\beta_j}| |z_{\beta_j}|^2) ( \sum_{j=5}^8 |d_{\beta_j}| |z_{\beta_j}|^2 )} \left( \sum_{j=5}^8 \sigma_j \, z_{\beta_j} \frac{d}{dz_{\beta_j}} \right) M (h-E -\sum_\alpha d_\alpha |z_{\alpha}|^2) \\ \lesssim \; &  \int d\mu (\bz) \frac{|z_\alpha|^2 (\sum_{j=1}^4 c_{\beta_j} |z_{\beta_j}|^2)}{ (\sum_{j=1}^4 |d_{\beta_j}| |z_{\beta_j}|^2) ( \sum_{j=5}^8 |d_{\beta_j}| |z_{\beta_j}|^2 )} \, M(h-E -\sum_\alpha d_\alpha |z_{\alpha}|^2) \\ &+ \sum_{j=5}^8 \int d\mu (\bz) |z_{\beta_j}| \left|\frac{d\phi (\bz)}{dz_{\beta_j}}\right| \, \frac{|z_\alpha|^2 (\sum_{j=1}^4 c_{\beta_j} |z_{\beta_j}|^2)}{ (\sum_{j=1}^4 |d_{\beta_j}| |z_{\beta_j}|^2) ( \sum_{j=5}^8 |d_{\beta_j}| |z_{\beta_j}|^2 )} \, M(h-E -\sum_\alpha d_\alpha |z_{\alpha}|^2) \,. \end{split} \]
With \[ M(h-E-\sum_\alpha d_\alpha |z_\alpha|^2) \lesssim \frac{1}{\sum_{j=1}^4 c_{\beta_j} |z_{\beta_j}|^2} \]
we find
\begin{equation}\label{eq:B6-Delta} |\text{B}_6| \lesssim \int d\mu (\bz) \frac{|z_\alpha|^2}{(\sum_{j=1}^4 |d_{\beta_j}| |z_{\beta_j}|^2 ) ( \sum_{j=5}^8 |d_{\beta_j}| |z_{\beta_j}|^2)} \left( 1 + \sum_{j=5}^8 |z_{\beta_j}| \left| \frac{d\phi (\bz)}{dz_{\beta_j}} \right| \right) \end{equation}
Therefore, using Lemma \ref{lm:pow}, Lemma \ref{lm:dphi}, Lemma \ref{lm:CS} and the fact that $|d_{\beta_j}| \geq 1/(2\Delta)$ for all $j=1,\dots ,8$ (see (\ref{eq:delta})), we find
\begin{equation}\label{eq:B6-2} \begin{split}
|\text{B}_6| \lesssim \; & \Delta^2 \, \left( \int d\mu (\bz) \frac{1}{\left( \sum_{j=1}^4 |z_{\beta_j}|^2 \right)^3} \right)^{1/3}  \left( \int d\mu (\bz) \frac{1}{\left( \sum_{j=5}^8 |z_{\beta_j}|^2 \right)^3} \right)^{1/3}  \left( \int d\mu (\bz) \, |z_{\alpha}|^6 \right)^{1/3} \\
&+ \Delta^2 \, \sum_{j=5}^8  \left( \int d\mu (\bz) \frac{1}{\left( \sum_{j=1}^4 |z_{\beta_j}|^2 \right)^3} \right)^{1/3}  \left( \int d\mu (\bz) \frac{1}{\left( \sum_{j=5}^8 |z_{\beta_j}|^2 \right)^3} \right)^{1/6} \\ & \hspace{5cm} \times \left(\int d\mu (\bz) \left| \frac{d\phi (\bz)}{dz_{\beta_j}} \right|^6 \right)^{1/6}\left( \int d\mu (\bz) \, |z_{\alpha}|^6 \right)^{1/3} \\
\lesssim \; & \Delta^2
\end{split} \end{equation}

As for the term $\text{B}_5$ on the r.h.s. of (\ref{eq:B5B6}), we integrate by parts.
We find
\begin{equation}\label{eq:B5-3terms} \begin{split}
\text{B}_5 = \; &\frac{1}{2}  \int d\mu (\bz) \left[ \sum_{j=1}^4 \sigma_j  \frac{d}{dz_{\beta_j}} \frac{z_{\beta_j} |z_\alpha|^2}{
\sum_{j=1}^4 |d_{\beta_j}| |z_{\beta_j}|^2} \right]  \left(\frac{\eps}{N} +  \sum_\gamma c_\gamma |z_{\gamma}|^2 \right)  \\ &\hspace{6cm} \times  \frac{1}{
(h - E -\sum_{\alpha}d_{\alpha}|z_{\alpha}|^2)^2 + (\frac{\eps}{N} + \sum_{\alpha}c_{\alpha}|z_{\alpha}|^2)^2}  \\
&+ \frac{1}{2} \int d\mu (\bz) \frac{|z_\alpha|^2 (\sum_{j=1}^4 \sigma_j  c_{\beta_j} |z_{\beta_j}|^2)}{\sum_{j=1}^4 |d_{\beta_j}| |z_{\beta_j}|^2}  \, \frac{1}{
(h - E -\sum_{\alpha}d_{\alpha}|z_{\alpha}|^2)^2 + (\frac{\eps}{N} + \sum_{\alpha}c_{\alpha}|z_{\alpha}|^2)^2}  \\
&+ \frac{1}{2} \sum_{j=1}^4 \sigma_j \int d\mu (\bz) \, z_{\beta_j} \frac{d\phi (\bz)}{dz_{\beta_j}} \, \frac{|z_\alpha|^2}{\sum_{j=1}^4 |d_{\beta_j}| |z_{\beta_j}|^2}  \left( \frac{\eps}{N} + \sum_{\gamma} c_{\gamma} |z_{\gamma}|^2 \right) \\ &\hspace{6cm} \times \frac{1}{
(h - E -\sum_{\alpha}d_{\alpha}|z_{\alpha}|^2)^2 + (\frac{\eps}{N} + \sum_{\alpha}c_{\alpha}|z_{\alpha}|^2)^2}
\end{split} \end{equation}
It follows easily from a bound similar to (\ref{eq:der-bd}) and proceeding then as in (\ref{eq:B1-1}) that
\begin{equation}\label{eq:B5-I}
\begin{split}
|\text{B}_5| & \lesssim \; \frac{1}{c_\alpha} \int d\mu (\bz) \frac{1}{\sum_{j=1}^4 |d_{\beta_j}| |z_{\beta_j}|^2} \left(1+ \sum_{j=1}^4 |z_{\beta_j}| \left| \frac{d\phi (\bz)}{dz_{\beta_j}} \right| \right) \lesssim \frac{\Delta}{c_\alpha} \end{split}\end{equation}
Alternatively, we can observe that
\begin{equation}\label{eq:alt-B5} \begin{split}
&\frac{1}{(h - E -\sum_{\alpha}d_{\alpha}|z_{\alpha}|^2)^2 + (\frac{\eps}{N} + \sum_\gamma c_\gamma |z_\gamma|^2 )^2} \\ &\hspace{0cm}=
\; - \frac{1}{\sum_{j=5}^8 |d_{\beta_j}| |z_{\beta_j}|^2} \left( \sum_{j=5}^8 \sigma_j z_{\beta_j} \frac{d}{dz_{\beta_j}} \right) L (h-E-\sum_{\alpha}d_{\alpha}|z_{\alpha}|^2) \\ &\hspace{0.5cm} - \frac{2 \left(\eps/N + \sum_\alpha c_\alpha |z_\alpha|^2 \right) \, (\sum_{j=5}^8 \sigma_j \, c_{\beta_j} |z_{\beta_j}|^2)}{\sum_{j=5}^8 |d_{\beta_j}| |z_{\beta_j}|^2} \int_{-\infty}^{h-E-\sum_\alpha d_\alpha |z_\alpha|^2} ds \frac{1}{\left(s^2 + (\frac{\eps}{N} + \sum_\alpha c_\alpha |z_\alpha|^2)^2 \right)^2}
\end{split} \end{equation}
where, as in (\ref{eq:L}), we set
\[  L(t) = \int_{-\infty}^t ds \frac{1}{s^2+ \left( \eps/ N + \sum_\gamma c_\gamma |z_\gamma|^2 \right)^2} \]
Inserting (\ref{eq:alt-B5}) into (\ref{eq:B5-3terms}), performing integration by parts (in the terms arising from the first line of (\ref{eq:alt-B5})), taking absolute values, using a bound similar to (\ref{eq:der-bd}) and the fact that \[ L(h-E-\sum_{\alpha} d_\alpha |z_\alpha|^2) \lesssim \frac{1}{\eps/N + \sum_{\gamma} c_{\gamma} |z_{\gamma}|^2} \] we conclude that
\begin{equation} \label{eq:B5-Delta} \begin{split}
|\text{B}_5| \lesssim \;  &\int d\mu (\bz) \frac{|z_\alpha|^2}{(\sum_{j=1}^4 |d_{\beta_j}| |z_{\beta_j}|^2 )(\sum_{j=5}^8 |d_{\beta_j}| |z_{\beta_j}|^2)} \\ &\hspace{3cm} \times  \left(1+ \sum_{j=1}^4 |z_{\beta_j}|^2 \left| \frac{d\phi (\bz)}{dz_{\beta_j}} \right|^2 + \sum_{j=1}^4 \sum_{i=5}^8 |z_{\beta_j}| |z_{\beta_i}| \, \left| \frac{d^2 \phi (\bz)}{dz_{\beta_j} dz_{\beta_i}} \right| \right) \\
\lesssim \; &\Delta^2 \int d\mu (\bz) \frac{|z_\alpha|^2}{(\sum_{j=1}^4 |z_{\beta_j}|^2 )(\sum_{j=5}^8 |z_{\beta_j}|^2)} \\ &\hspace{3cm} \times  \left(1+ \sum_{j=1}^4 |z_{\beta_j}|^2 \left| \frac{d\phi (\bz)}{dz_{\beta_j}} \right|^2 + \sum_{j=1}^4 \sum_{i=5}^8 |z_{\beta_j}| |z_{\beta_i}| \, \left| \frac{d^2 \phi (\bz)}{dz_{\beta_j} dz_{\beta_i}} \right| \right)
\end{split}
\end{equation}
Therefore, we obtain
\begin{equation}\label{eq:B5-2}\begin{split}
|\text{B}_5| \lesssim \; & \Delta^2 \left( \int d\mu (\bz) \frac{1}{\left( \sum_{j=1}^4 |z_{\beta_j}|^2 \right)^3} \right)^{1/3} \left( \int d\mu (\bz) \frac{1}{\left( \sum_{j=5}^8 |z_{\beta_j}|^2 \right)^3} \right)^{1/3} \left( \int d\mu (\bz) \, |z_\alpha|^6 \right)^{1/3} \\
&+ \Delta^2 \sum_{j=1}^4 \left( \int d\mu (\bz) \left| \frac{d\phi (\bz)}{dz_{\beta_j}} \right|^6 \right)^{1/3} \left( \int d\mu (\bz) \frac{1}{\left( \sum_{j=5}^8 |z_{\beta_j}|^2 \right)^3} \right)^{1/3} \left( \int d\mu (\bz) \, |z_\alpha|^6 \right)^{1/3} \\
&+ \Delta^2 \sum_{j=1}^4 \sum_{i=5}^8 \left( \int d\mu (\bz) \frac{1}{\left( \sum_{j=1}^4 |z_{\beta_j}|^2 \right)^3} \right)^{1/6} \left( \int d\mu (\bz) \frac{1}{\left( \sum_{j=5}^8 |z_{\beta_j}|^2 \right)^3} \right)^{1/6} \\ &\hspace{6cm} \times  \left( \int d\mu (\bz) \left| \frac{d^2\phi (\bz)}{dz_{\beta_j} dz_{\beta_i}} \right|^2 \right)^{1/2}\left( \int d\mu (\bz) |z_\alpha|^{12} \right)^{1/6}  \\
\lesssim \; & \Delta^2
\end{split} \end{equation}
by Lemma \ref{lm:pow}, Lemma \ref{lm:dphi} and Lemma \ref{lm:CS}. Together with (\ref{eq:B5-I}), (\ref{eq:B6-I}), (\ref{eq:B6-2}), we obtain the last two bounds on the r.h.s. of (\ref{eq:propII}).

\medskip

In order to show (\ref{eq:propIII}), we proceed as in the proof of the bound $\Delta^2$ for the l.h.s. of (\ref{eq:propII}) (notice that the only difference between the l.h.s. of (\ref{eq:propII}) and (\ref{eq:propIII}) is the factor $|z_{\alpha}|^2$, which, however, did not play any role in the proof of the bound proportional to $\Delta^2$ on the r.h.s. of (\ref{eq:propII})). We write
\[ \text{C} := \int  d\mu (\bz)  \, \left(\frac{\eps}{N} + \sum_\gamma c_\gamma |z_{\gamma}|^2\right)   \frac{h - E -\sum_{\alpha}d_{\alpha}|z_\alpha|^2}{\left[(h - E -\sum_{\alpha}d_{\alpha}|z_{\alpha}|^2)^2 + (\frac{\eps}{N} + \sum_{\alpha}c_{\alpha}|z_{\alpha}|^2)^2\right]^2} \, . \]
We decompose $\text{C}$, similarly to (\ref{eq:B5B6}), as
\[ \begin{split} \text{C} =  \; &-\frac{1}{2} \int d\mu (\bz)  \,  \left(\frac{\eps}{N} +  \sum_\gamma c_\gamma |z_{\gamma}|^2 \right) \, \frac{1}{\sum_{j=1}^4 |d_{\beta_j}| |z_{\beta_j}|^2} \\ &\hspace{2cm} \times  \left( \sum_{j=1}^4 \sigma_j \, z_{\beta_j} \frac{d}{dz_{\beta_j}} \right) \, \frac{1}{
(h - E -\sum_{\alpha}d_{\alpha}|z_{\alpha}|^2)^2 + (\frac{\eps}{N} + \sum_{\alpha}c_{\alpha}|z_{\alpha}|^2)^2}  \\
&-\,\int d\mu (\bz) \, \frac{\left( \sum_{j=1}^4 \sigma_j c_{\beta_j} |z_{\beta_j}|^2 \right)}{\sum_{j=1}^4 |d_{\beta_j}| |z_{\beta_j}|^2}  \frac{\left(\frac{\eps}{N} + \sum_\gamma c_\gamma |z_{\gamma}|^2\right)^2 }{\left[(h - E -\sum_{\alpha}d_{\alpha}|z_{\alpha}|^2)^2 + (\frac{\eps}{N} + \sum_{\alpha}c_{\alpha}|z_{\alpha}|^2)^2\right]^2}
\\ =: \; & \text{C}_1 + \text{C}_2   \, .
\end{split} \]
Analogously to (\ref{eq:B6-Delta}), we obtain
\[ |\text{C}_1| \lesssim \int d\mu (\bz) \frac{1}{(\sum_{j=1}^4 |d_{\beta_j}| |z_{\beta_j}|^2 ) ( \sum_{j=5}^8 |d_{\beta_j}| |z_{\beta_j}|^2)} \left( 1 + \sum_{j=1}^4 |z_{\beta_j}| \left| \frac{d\phi (\bz)}{dz_{\beta_j}} \right| \right) \, . \]
Analogously to (\ref{eq:B5-Delta}), we find
\[\begin{split}  |\text{C}_2| \lesssim \; &\int d\mu (\bz) \frac{1}{(\sum_{j=1}^4 |d_{\beta_j}| |z_{\beta_j}|^2 )(\sum_{j=5}^8 |d_{\beta_j}| |z_{\beta_j}|^2)} \\ &\hspace{3cm} \times  \left(1+ \sum_{j=1}^4 |z_{\beta_j}|^2 \left| \frac{d\phi (\bz)}{dz_{\beta_j}} \right|^2 + \sum_{j=1}^4 \sum_{i=5}^8 |z_{\beta_j}| |z_{\beta_i}| \, \left| \frac{d^2 \phi (\bz)}{dz_{\beta_j} dz_{\beta_i}} \right| \right) \end{split} \]
Hence, we obtain
\[ \begin{split}  |\text{C}| \lesssim \; &\Delta^2 \, \int d\mu (\bz) \frac{1}{(\sum_{j=1}^4  |z_{\beta_j}|^2 )(\sum_{j=5}^8 |z_{\beta_j}|^2)} \\ &\hspace{3cm} \times  \left(1+ \sum_{j=1}^4 |z_{\beta_j}|^2 \left| \frac{d\phi (\bz)}{dz_{\beta_j}} \right|^2 + \sum_{j=1}^4 \sum_{i=5}^8 |z_{\beta_j}| |z_{\beta_i}| \, \left| \frac{d^2 \phi (\bz)}{dz_{\beta_j} dz_{\beta_i}} \right| \right) \end{split} \]
Similarly to (\ref{eq:B5-2}), we find $|\text{C}| \lesssim \Delta^2$.
This completes the proof of (\ref{eq:propIII}).
\end{proof}

\begin{lemma}\label{lm:CS}
Let the probability density $h$ be such that (\ref{eq:ass2}) is satisfied, and let the measure $d\mu (\bz)$ be as in (\ref{eq:dmu}). Let $m\in \bN$ and $p \in \bR$, with $0 < p < m$. For any indices $\beta_1, \dots , \beta_m \in \{ 1,2, \dots , N-1\}$, we have
\[ \int d\mu (\bz) \, \frac{1}{\left( \sum_{j=1}^m |z_{\beta_j}|^2 \right)^p} \lesssim \int \left| \frac{h' (s)}{h(s)} \right|^{2\overline{p}} h(s) ds \]
where $\overline{p} \in \bN$ is the smallest integer larger than $p$.
\end{lemma}

\begin{proof}
Observe that
\[ \sum_{j=1}^m \frac{d}{dz_{\beta_j}} \frac{z_{\beta_j}}{\left( \sum_{i=1}^m |z_{\beta_i}|^2 \right)^p} = (m-p) \, \frac{1}{\left( \sum_{i=1}^m |z_{\beta_i}|^2 \right)^p} \]
Therefore, recalling from (\ref{eq:phi}) that $d\mu (\bz) = e^{-\phi (\bz)} d\bz d \overline{\bz}$, we find
\[ \begin{split} \text{I} :=  \; & \int d\mu (\bz) \, \frac{1}{\left( \sum_{j=1}^m |z_{\beta_j}|^2 \right)^p}\\ =  \; & \frac{1}{m-p} \sum_{j=1}^m  \int d\mu (\bz) \frac{d}{dz_{\beta_j}} \frac{z_{\beta_j}}{\left( \sum_{i=1}^m |z_{\beta_i}|^2 \right)^p} \\
= \; & \frac{1}{m-p} \sum_{j=1}^m \int d\mu (\bz) \frac{d\phi (\bz)}{dz_{\beta_j}} \, \frac{z_{\beta_j}}{\left( \sum_{i=1}^m |z_{\beta_i}|^2 \right)^p} \end{split} \]
Hence, H\"older inequality implies that
\[ \begin{split} \text{I} \leq \; & \frac{1}{m-p} \sum_{j=1}^m \left( \int d\mu (\bz) \, \left| \frac{d\phi (\bz)}{dz_{\beta_j}} \right|^{2p} \right)^{\frac{1}{2p}} \left( \int d\mu (\bz) \, \left( \frac{|z_{\beta_j}|}{\left( \sum_{i=1}^m |z_{\beta_i}|^2 \right)^p} \right)^{\frac{2p}{2p-1}} \right)^{1-\frac{1}{2p}}  \\
\leq \; &\frac{\text{I}^{1-\frac{1}{2p}}}{m-p} \sum_{j=1}^m \left( \int d\mu (\bz) \, \left| \frac{d\phi (\bz)}{dz_{\beta_j}} \right|^{2p} \right)^{\frac{1}{2p}}\end{split} \]
It follows from Lemma \ref{lm:dphi} that
\[ \begin{split}
\text{I} \leq \left(\frac{m}{m-p} \right)^{2p} \, \sup_j  \int d\mu (\bz) \, \left| \frac{d\phi (\bz)}{dz_{\beta_j}} \right|^{2p} \lesssim 1+ \int \left| \frac{h' (s)}{h(s)} \right|^{2\overline{p}} \, h(s) ds \, . \end{split} \]
\end{proof}

\begin{lemma}\label{lm:dphi}
Let the probability density $h$ be such that (\ref{eq:ass2}) is satisfied, let the measure $d\mu (\bz)$ be as in (\ref{eq:dmu}), and let $\phi (\bz)$ be as in (\ref{eq:phi}) . For any $m \in \bN$, there exists a constant $C_m$ such that
\[ \int d\mu (\bz) \, \left| \frac{d\phi (\bz)}{dz_{\beta}} \right|^{2m} \leq C_m  \int \left| \frac{h' (s)}{h(s)} \right|^{2m} \, h(s) ds \]
for any index $\beta \in \{1, \dots, N-1\}$.
Moreover,
\[ \int d\mu (\bz) \left| \frac{d^2 \phi (\bz)}{dz_{\beta_1} dz_{\beta_2}} \right|^2 \lesssim \int \left| \frac{h'' (s)}{h(s)} \right|^2 \, h(s) ds + \int \left| \frac{h' (s)}{h(s)} \right|^4 \, h (s) ds \]
for any indices $\beta_1, \beta_2 \in \{1, \dots, N-1\}$
\end{lemma}

\begin{proof}
{F}rom (\ref{eq:phi}), we have (recall that $g=-\log h$),  \[ \phi (\bz) = \sum_{\ell=1}^{N-1} g (\text{Re } (U\bz)_\ell) + g (\text{Im } (U\bz)_\ell) \] Hence
\[ \frac{d\phi (\bz)}{dz_\beta} = \frac{1}{2} \sum_{\ell=1}^{N-1} U_{\ell, \beta} \left( g' (\text{Re } (U\bz)_{\ell} ) -i  g' (\text{Im } (U\bz)_{\ell} ) \right) \]
Therefore
\[ \begin{split}
\int d\mu (\bz) &\left| \frac{d\phi (\bz)}{dz_{\beta}} \right|^{2m} \\ = \; & \frac{1}{2^m} \sum_{\ell_1, \dots , \ell_{2m} =1}^{N-1} U_{\ell_1,\beta} \dots U_{\ell_m ,
\beta} \overline{U_{\ell_{m+1},\beta}} \dots \overline{U_{\ell_{2m},\beta}} \\
&\times \int d\mu (\bz) \, \prod_{j=1}^m \left( g' (\text{Re } (U\bz)_{\ell_j} ) -i  g' (\text{Im } (U\bz)_{\ell_j } ) \right) \prod_{j=m+1}^{2m} \left( g' (\text{Re } (U\bz)_{\ell_j} ) + i  g' (\text{Im } (U\bz)_{\ell_j} ) \right)
\end{split} \]
The integral vanishes if there exists an index $\ell_j$ such that $\ell_j \not = \ell_i$ for all $i \not = j$. Hence, changing coordinates to $b_\ell = (U\bz)_\ell$, we find
\[ \begin{split}
\int d\mu (\bz) \left| \frac{d\phi (\bz)}{dz_{\beta}} \right|^{2m} \lesssim \; &  \sum_{r=1}^m \sum_{\alpha_1, \dots , \alpha_r \geq 2} {\bf 1} (\alpha_1 + \dots + \alpha_r = 2m )   |U_{\ell_1, \beta}|^{\alpha_1} \dots |U_{\ell_r, \beta}|^{\alpha_r} \\ & \hspace{.5cm} \times\int \prod_{j=1}^{N-1} h (\text{Re } b_j) h (\text{Im } b_j) \, \prod_{n=1}^r \left( |g' (\text{Re } b_{\ell_n})| + |g' (\text{Im } b_{\ell_n})| \right)^{\alpha_n} \\ \lesssim \; &  \sum_{r=1}^m \sum_{\alpha_1, \dots , \alpha_r \geq 2}  {\bf 1} (\alpha_1 + \dots + \alpha_r = 2m ) \,  |U_{\ell_1, \beta}|^{\alpha_1} \dots |U_{\ell_r,\beta}|^{\alpha_r} \\ & \hspace{.5cm} \times \sum_{n=1}^r  \int h(\text{Re } b_{\ell_n}) \, h (\text{Im } b_{\ell_n}) \, \left( |g' (\text{Re } b_{\ell_n})|^{2m} + |g' (\text{Im } b_{\ell_n})|^{2m} \right) \\ \lesssim \; & \int \left| \frac{h' (s)}{h(s)} \right|^{2m} \, h (s) ds
\end{split} \]
where we used the fact that, for any $\alpha \geq 2$, $\sum_{\ell} |U_{\ell,\beta}|^\alpha \leq 1$. On the other hand, we have,
\[  \frac{d^2 \phi (\bz)}{dz_{\beta_1} dz_{\beta_2}} = \frac{1}{4} \sum_{\ell=1}^{N-1} U_{\ell,\beta_1} U_{\ell,\beta_2} \left( g'' (\text{Re } (U\bz)_\ell) - g'' (\text{Im } (U \bz)_\ell) \right) \, . \]
Applying Cauchy-Schwarz inequality, we find
\[ \begin{split}
\int d\mu (\bz) \left|  \frac{d^2 \phi (\bz)}{dz_{\beta_1} dz_{\beta_2}} \right|^2 \lesssim \; &  \sum_{\ell_1, \ell_2=1}^{N-1} |U_{\ell_1,\beta_1}|^2  |U_{\ell_2,\beta_1}|^2  \int d\mu (\bz) \left( |g'' (\text{Re } (U\bz)_{\ell_1})|^2 + |g'' (\text{Im } (U\bz)_{\ell_1})|^2  \right) \\ &+ \sum_{\ell_1, \ell_2=1}^{N-1} |U_{\ell_1,\beta_2}|^2  |U_{\ell_2,\beta_2}|^2  \int d\mu (\bz) \left( |g'' (\text{Re } (U\bz)_{\ell_2})|^2 + |g'' (\text{Im } (U\bz)_{\ell_2})|^2  \right)
\\ \lesssim \; &\int |g'' (s)|^2 h(s) ds \end{split} \]
The lemma follows from because \[ g'' (s) = \frac{h'' (s)}{h(s)} - \frac{(h' (s))^2}{h(s)^2} \,. \]
\end{proof}

\begin{lemma}\label{lm:pow}
Let the measure $d\mu (\bz)$ be as in (\ref{eq:dmu}). For any $m \in \bN$, there exists a constant $C_m > 0$ such that
\[ \int d\mu (\bz) \, |z_\alpha|^m \leq C_m \]
for every index $\alpha \in \{ 1,\dots , N-1 \}$.
\end{lemma}

\begin{proof}
Note that, with $\bb = U\bz$, we have
\[ \int d\mu (\bz) |z_\alpha|^m = \int \prod_{\ell=1}^{N-1} h (\text{Re } b_j) \, h (\text{Im } b_j) \, |\bb \cdot \bu_\alpha|^m = \E \, |\bb \cdot \bu_\alpha|^m. \]
(Recall the notation $\xi_\alpha = |\bb \cdot \bu_\alpha|^2$, introduced after (\ref{eq:xii})). {F}rom Proposition 4.5 in \cite{ESY3}, we conclude that
\[ \P \left( |\bb \cdot \bu_\alpha|^2 \geq K \right) \lesssim e^{-cK} \]
Therefore,
\[ \E \, |\bb \cdot \bu_\alpha|^m = m \int_0^\infty dx \, x^{m-1} \, \P \left( |\bb \cdot \bu_\alpha| \geq x \right) \lesssim \int_0^\infty dx \, x^{m-1} e^{-cx} < \infty \, . \]
\end{proof}

\end{document}